\documentclass[11pt]{article}
\usepackage[latin9]{inputenc}
\usepackage{array}
\usepackage{longtable}
\usepackage{amsmath}
\usepackage{amsthm}
\usepackage{amssymb}
\usepackage{cancel}
\usepackage{graphicx}
\usepackage{epstopdf}
\usepackage[square, authoryear]{natbib}
\makeatletter

\providecommand{\tabularnewline}{\\}

\newcommand{\lyxaddress}[1]{
\par {\raggedright #1
\vspace{1.4em}
\noindent\par}
}
\theoremstyle{plain}
\newtheorem{thm}{\protect\theoremname}
  \theoremstyle{definition}
  \newtheorem{defn}[thm]{\protect\definitionname}
  \theoremstyle{plain}
  \newtheorem{lem}[thm]{\protect\lemmaname}

\usepackage{MnSymbol}
\usepackage{calrsfs}
\usepackage{bbm}

\usepackage{color}

\newcommand{\twopartdef}[4]{
\left\{
\begin{array}{ll}
#1 & \mbox{if } #2 \\
#3 & \mbox{if } #4
\end{array}
\right.
}

\@ifundefined{showcaptionsetup}{}{%
 \PassOptionsToPackage{caption=false}{subfig}}
\usepackage{subfig}
\makeatother

  \providecommand{\definitionname}{Definition}
  \providecommand{\lemmaname}{Lemma}
\providecommand{\theoremname}{Theorem}

\begin{document}

%\begin{frontmatter}{}

\title{{Centralities in Simplicial Complexes. Applications to Protein Interaction
Networks}}
%}

%\address{
\author{Ernesto Estrada and Grant Ross}
\maketitle

\lyxaddress{Department of Mathematics and Statistics, University of Strathclyde, 26 Richmond Street, Glasgow G1 1HX, UK}

\begin{abstract}
Complex networks can be used to represent complex systems which originate
in the real world. Here we study a transformation of these complex
networks into simplicial complexes, where cliques represent the simplices
of the complex. We extend the concept of node centrality to that of
simplicial centrality and study several mathematical properties of
degree, closeness, betweenness, eigenvector, Katz, and subgraph centrality
for simplicial complexes. We study the degree distributions of these
centralities at the different levels. We also compare and describe
the differences between the centralities at the different levels.
Using these centralities we study a method for detecting essential
proteins in PPI networks of cells and explain the varying abilities
of the centrality measures at the different levels in identifying
these essential proteins.
\end{abstract}

%\end{frontmatter}{}

\section{Introduction}

There is little doubt that the use of graphs and networks to represent
the skeleton of complex systems has been a successful paradigm. This
simple representation in which nodes of the graph accounts for the
entities of a complex system and the edges describe the interactions
between these entities captures many of the complex structural and
dynamical properties of the represented systems. However, such representation
is far from complete. One of its main drawbacks is its concentration
of binary relations only. That is, in a network the interaction between
entities occurs in a pairwise way. This excludes other higher-order
interactions involving groups of entities. Let us provide some examples.
Networks have been widely used to represent protein-protein interactions
(PPIs) where the nodes represent proteins and pairs of interacting
proteins are connected by edges of the network. These PPI networks
contain many triangles in which triples of proteins are considered
to be interacting to each other. Now, let us consider that there are
three proteins $A$, $B$ and $C$ that form a heterotrimer, that
is an $ABC$ complex in which the three proteins interact with each
other at the same time. The network-theoretic representation is not
able to differentiate this situation from the case where there are
three proteins $A$, $B$ and $C$ and they interact in a pairwise
manner, e.g. $AB$, $AC$, $BC$. The existence of heterotrimers in
well-documented, an example is the heterotrimeric $G$ protein formed
by the three proteins $G_{\alpha}$, $G_{\beta}$ and $G_{\gamma}$.
An attempt to amend this problem has been made by using hypergraphs,
also known as hypernetworks. In this case, the triple of proteins
form an hyper-edge which accounts for the simultaneous interaction
of the proteins in the complex. However, hypergraphs have a main drawback
when trying to capture all the subtleties of these complexes. For
instance, in the heterotrimeric $G$ protein, the proteins $G_{\beta}$
and $G_{\gamma}$ form a subcomplex known as $G_{\beta\gamma}$ which
is part of the trimeric form. This situation is not necessarily captured
by the hypergraph representation where hyperedges are not necessarily
closed under the subset operation. Examples of real-world systems
where this closure under the subset operation is required abound and
a very nice example provided by Maleti\'{c} and Rajkovi\'{c} \citep{maletic2012combinatorial}
according to them provided by Spivak\textemdash , where four people
have a chat in which everybody can listen to each other. Obviously,
the conversation is not pairwise as represented by the graph, and
is not only in the form of the hyper-edge represented by the hypergraph,
but a combination of the quadruple, triangles and edges. The best
way to represent such situations is by means of the so-called \textit{simplicial
complexes}.

Informally, a simplicial complex is a mathematical object, which originated
in algebraic topology and is a generalization of a network. Starting
with a set of nodes, instead of being limited to sets of size two,
the simplices can contain any number of nodes. A characteristic feature
of a simplex $S$ of a certain size is that all subsets of $S$ must
also be simplices. In this way simplicial complexes differ from hypergraphs.
For instance, if there is a simplex $\{1,3,4,6\}$ in a simplicial
complex then $\{1,3,4\},\{1,3,6\},\{1,4,6\},\{3,4,6\}$ are also simplices
in the simplicial complex. All subsets of those four simplices must
also be simplices in the complex. There is a recent interest in these
mathematical objects for representing complex systems and we should
mention here their applications to study brain networks \citep{giusti2016two,courtney2016generalized,lee2012persistent,petri2014homological,pirino2015topological},
social systems \citep{maletic2009simplicial} \citep{maletic2014consensus,kee2016information},
biological networks \citep{xia2014persistent} \citep{xia2015multidimensional,cang2015topological},
and infrastructural systems \citep{muhammad2006control} \citep{tahbaz2010distributed,de2007coverage,de2005blind,ghrist2005coverage}.

Centrality indices have been among the most successful tools used
for discovering structural and dynamical properties of networks. A
centrality index is a numeric quantification of the 'importance' of
a node in terms of its position, structural and/or dynamical, in the
network. Here, we extend this concept to simplicial complexes to capture
the relevance of a simplex of a given order in a simplicial complex.
In particular, we apply this extended concept to the study of properties
of protein-protein interaction (PPI) networks.

\section{Preliminaries}

Simplicial complexes have been much studied in the literature \citep{horak2009persistent,sizemore2016classification}
and definitions similar to those which appear in the preliminaries
section can be found elsewhere \citep{muhammad2006control,tahbaz2010distributed,muhammad2007decentralized,maletic2012combinatorial,goldberg2002combinatorial}.
However, we repeat them here to make this paper self-contained.

Let $V$ be a set of nodes or vertices. Then a $k$-simplex is a set
$\{v_{0},v_{1},\dots,v_{k}\}$ such that $v_{i}\in V$ and $v_{i}\neq v_{j}$
for all $i\neq j$. A face of a $k$-simplex is a $(k-1)$-simplex
of the form $\{v_{0},\dots,v_{i-1},v_{i+1},\dots,v_{k}\}$ for $0\leq i\leq k$.
A simplicial complex $C$ is a collection of simplices such that if
a simplex $S$ is a member of $C$ then all faces of $S$ are also
members of $C$. Less formally, a simplicial complex is a collection
of simplices such that if $\{v_{0},v_{1},$$\dots,v_{k}\}$ is a simplex
then all of its faces $\{v_{0},\dots,v_{i-1},v_{i+1},\dots,v_{k}\}$
are also simplices, and all of the faces of its faces $\{v_{0},\dots,v_{i-1},v_{i+1},\dots,v_{j-1},v_{j+1},\dots,v_{k}\}$
are also simplices, and so on down to the $0$-simplices, which are
formed just by the nodes. As mentioned previously, networks can be
realized as simplicial complexes. The nodes are the $0$-simplices
which are specified by the set $V$, while the edges are the $1$-simplices
and there are no higher order simplices. It is also possible to create
simplicial complexes from networks. In this work we will be interested
only in the kind of simplicial complexes defined below, which are
known as \textit{clique complexes}. A clique complex is a simplicial
complex formed from a network as follows. The nodes of the network
become the nodes of the simplicial complex. Let $X$ be a clique of
$k$ nodes in the network. Then, $X$ is a $(k-1)$-simplex in the
clique complex. As an example in Figure \ref{Simplicla complex1}
we illustrate a simplicial complex which has one $3$-simplex $\{1,2,3,4\}$,
seven $2$-simplices $\{1,2,3\},$$\{1,2,4\},$$\{1,3,4\},$$\{2,3,4\},$$\{3,4,5\},$$\{4,5,6\}$
and$\{$$6,$$7,$$8\}$. It also has fourteen $1$-simplices represented
by the edges and nine $0$-simplices which are usually known as the
nodes.

\begin{figure}
\begin{centering}

\includegraphics[width=0.65\textwidth]{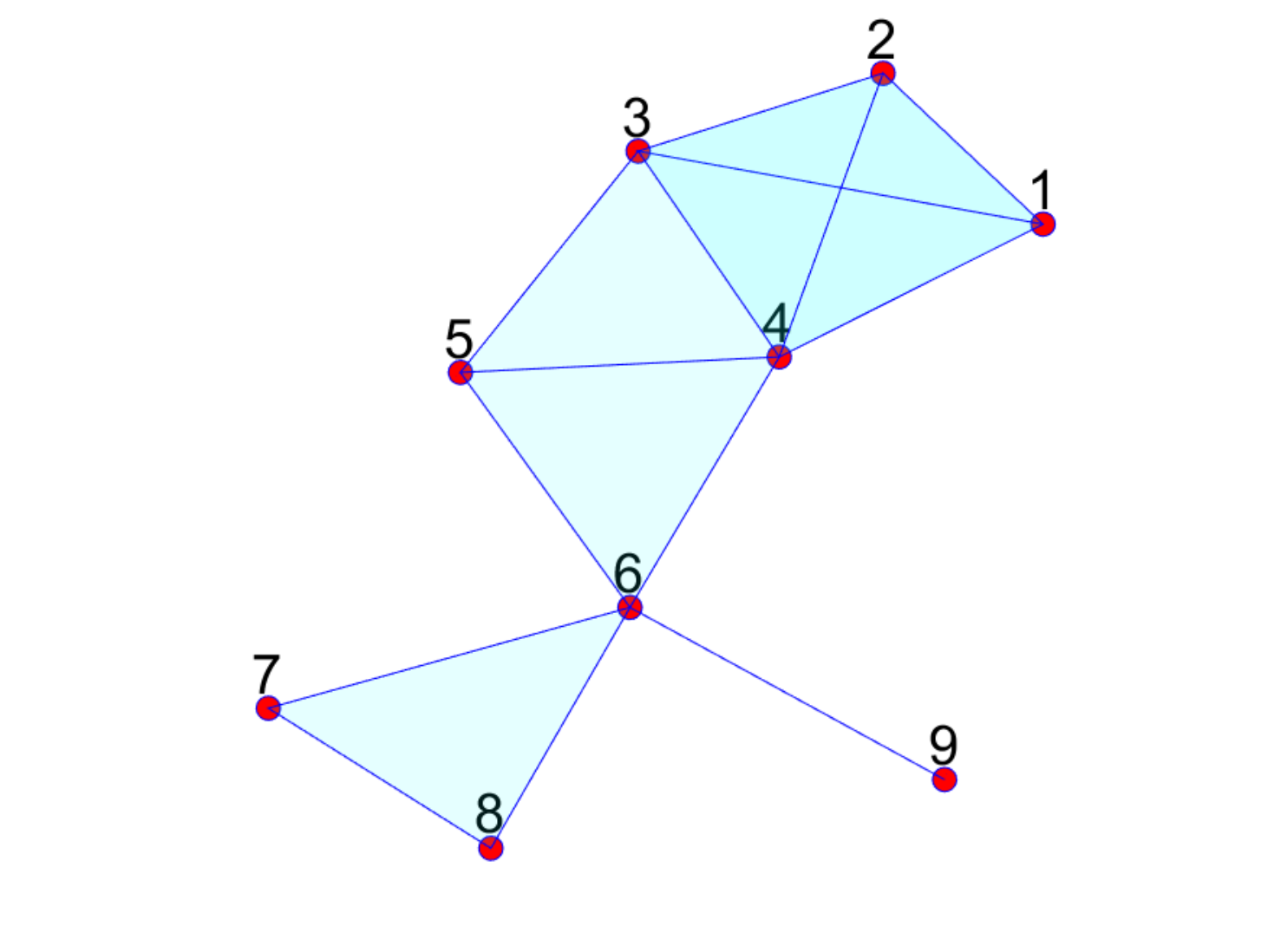}
\par\end{centering}
\caption{A simplicial complex with labeled nodes.}

\label{Simplicla complex1}
\end{figure}

In network theory it is fairly clear when two nodes are adjacent.
However, adjacency is less easy to define in simplicial complexes.
There are two ways in which two $k$-simplices $\sigma_{j}$ and $\sigma_{i}$
can be considered to be adjacent. We call them \textit{lower} and
\textit{upper adjacency}. Let $\sigma_{j}$ and $\sigma_{i}$ be two
$k$-simplices. Then, the two $k$-simplices are lower adjacent if
they share a common face. That is, for two distinct $k$-simplices
$\sigma_{j}=\{v_{0},v_{1},\dots,v_{k}\}$ and $\sigma_{i}=\{w_{0},w_{1},\dots,w_{k}\}$
then $\sigma_{j}$ and $\sigma_{i}$ are lower adjacent if and only
if there is a $(k-1)$-simplex $\tau=\{x_{0},x_{1},\dots,x_{k-1}\}$
such that $\tau\subset\sigma_{j}$ and $\tau\subset\sigma_{i}$. We
denote lower adjacency by $\sigma_{j}\smile\sigma_{i}$. For instance,
in the simplicial complex in Figure \ref{Simplicla complex1}, the
$1$-simplices $\{6,7\}$ and $\{6,9\}$ are lower adjacent because
the $0$-simplex $\{6\}$ is a common face of them and we can write
$\{6,7\}\smile\{6,9\}$. Similarly, $\{1,3,4\}\smile\{3,4,5\}$ are
lower adjacent as they share the common face $\{3,4\}$. However,
$\{4,5,6\}$ and $\{6,7,8\}$ are not lower adjacent because although
they have the common $0$-simplex $\{6\}$ they would need to share
a common $1$-simplex to be considered lower adjacent. Note that two
$0$-simplices can never be lower adjacent as we do not allow $\emptyset$
to be a $-1$-simplex. Let $\sigma_{j}$ and $\sigma_{i}$ be two
$k$-simplices. Then,the two $k$-simplices are upper adjacent if
they are both faces of the same common $(k+1)$-simplex. That is,
for $\sigma_{j}=\{v_{0},v_{1},\dots,v_{k}\}$ and $\sigma_{i}=\{w_{0},w_{1},\dots,w_{k}\}$
then $\sigma_{j}$ and $\sigma_{i}$ are upper adjacent if and only
if there is a $(k+1)$-simplex $\tau=\{x_{0},x_{1},\dots,x_{k+1}\}$
such that $\sigma_{j}\subset\tau$ and $\sigma_{i}\subset\tau$. We
denote the upper adjacency by $\sigma_{j}\frown\sigma_{i}$. In the
simplicial complex in Figure \ref{Simplicla complex1}, the $1$-simplices
$\{5,6\}$ and $\{4,6\}$ are upper adjacent because they are both
faces of the $2$-simplex $\{4,5,6\}$which is a common face of them.
So we can write $\{5,6\}\frown\{4,6\}$. Similarly, $\{1,3,4\}\frown\{2,3,4\}$
are upper adjacent as they are both faces of the $3$-simplex $\{1,2,3,4\}$.
However, $\{4,5,6\}$ is not upper adjacent to any other simplex as
it is not part of any $3$-simplices. Also note that $\{6\}\frown\{7\}$
are upper adjacent because they are both being faces of $\{6,7\}$.
So two $0$-simplices are upper adjacent if they are both faces of
a $1$-simplex which is identical to saying that two nodes are adjacent
if they are connected by an edge in the network theoretic sense. Hence
upper adjacency of $0$-simplices is the same as network theoretic
adjacency.

We shall now introduce some families of simplicial complexes which
shall be important later in the paper. Firstly, we introduce the family
denoted $S_{l}^{k}$. The simplicial complex $S_{l}^{k}$ consists
of a central $(k-1)$-simplex which is a face of every one of the
$l$ $k$-simplices. In addition, there are no other simplices except
those necessary by the closure axiom. For instance, $S_{l}^{2}$ would
consist of an edge $\{1,2\}$ and $l$ triangles of the form ${1,2,i}$
in addition to all subsimplices necessary by the closure axiom. While,
$S_{l}^{1}$ consists of a central node with $l$ pendant nodes connected
to it, which corresponds to the star graph in graph theory. The simplicial
complex $S_{5}^{2}$ is shown in Figure \ref{Families}(left).

Next we introduce a family of simplicial complexes labeled $t^{k}(x_{1},x_{2},\dots,x_{k+1})$
which consists of a central $k$-simplex with $x_{1}$ $k$-simplices
lower adjacent through one face, $x_{2}$ $k$-simplices lower adjacent
through another, and so on. A $k$-simplex which is lower adjacent
to the central $k$-simplex can only be lower adjacent to other $k$-simplices
which are lower adjacent to the central $k$-simplex through the same
face as itself. There are no other simplices except those necessary
by the closure axiom. One member of this family of simplices, $t^{2}(1,2,4)$
is shown in Figure \ref{Families}(center).

The final family of simplicial complexes which we shall introduce
are denoted $P_{l}^{k}$, consisting of a $k$-simplex at one end
which is only adjacent to one other $k$-simplex. This one is only
lower adjacent to the first $k$-simplex and another $k$-simplex,
and so on until arriving at another end $k$-simplex. In addition,
there are $l$ $k$-simplices in the simplicial complex and no other
simplices except those necessary by the closure axiom. Note that a
simplicial complex $P_{l}^{1}$ is the same as a path graph in the
traditional network theory. The simplicial complex$P_{5}^{2}$ is
illustrated in Figure \ref{Families}(right).

\begin{figure}
\subfloat[]{\begin{centering}
\includegraphics[width=0.33\textwidth]{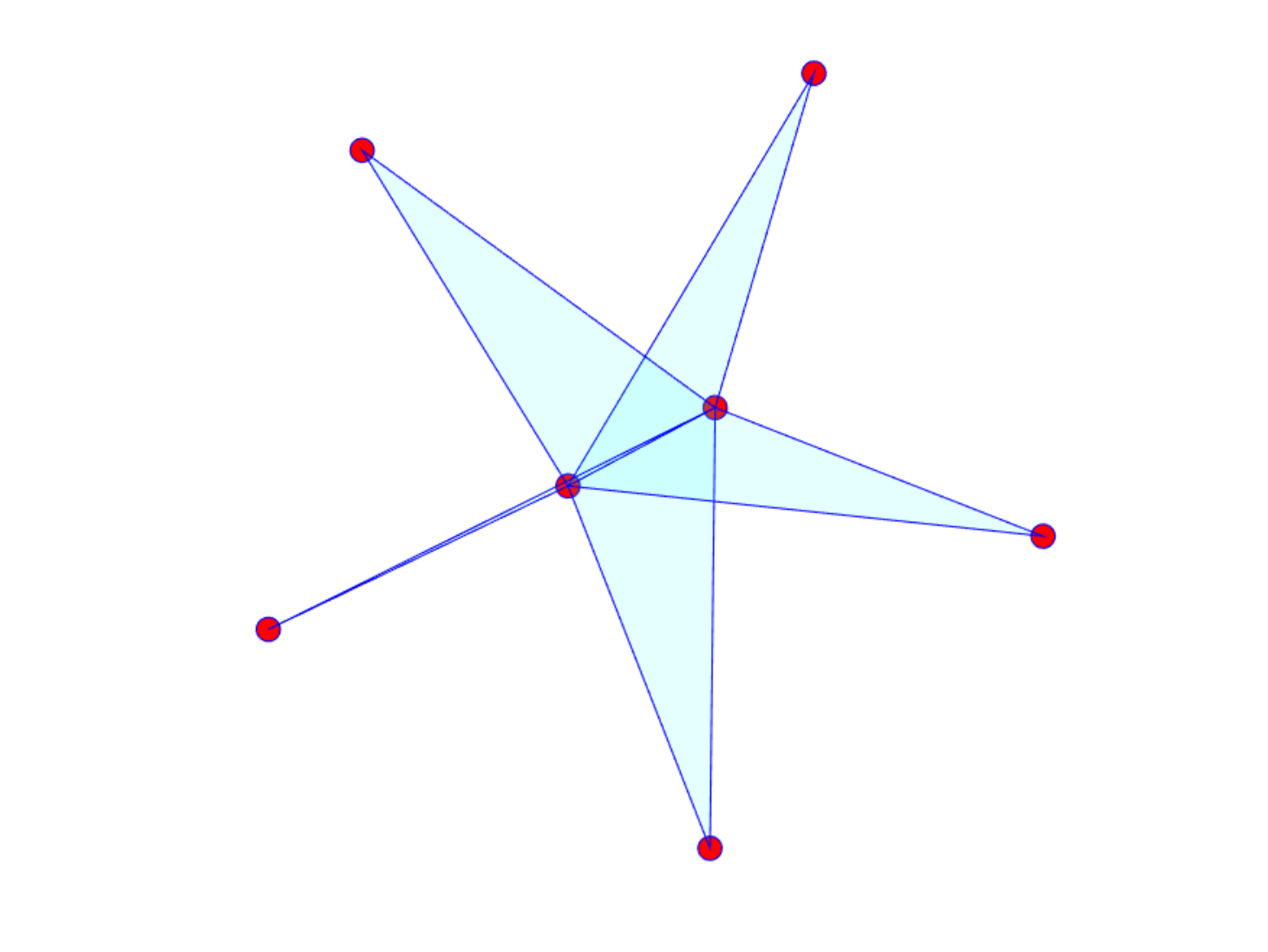}
\par\end{centering}
}\subfloat[]{\begin{centering}
\includegraphics[width=0.33\textwidth]{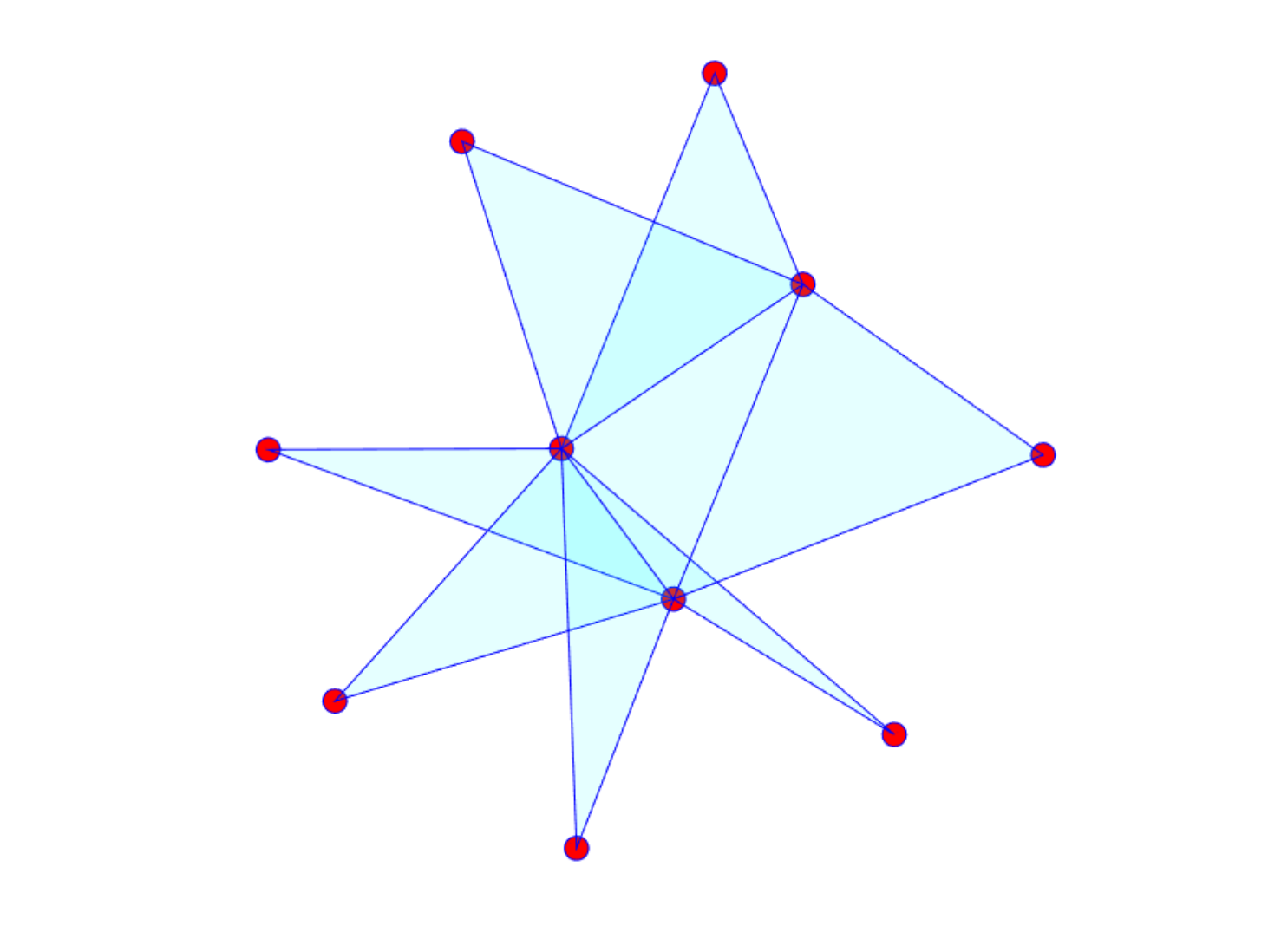}
\par\end{centering}
}\subfloat[]{\begin{centering}
\includegraphics[width=0.33\textwidth]{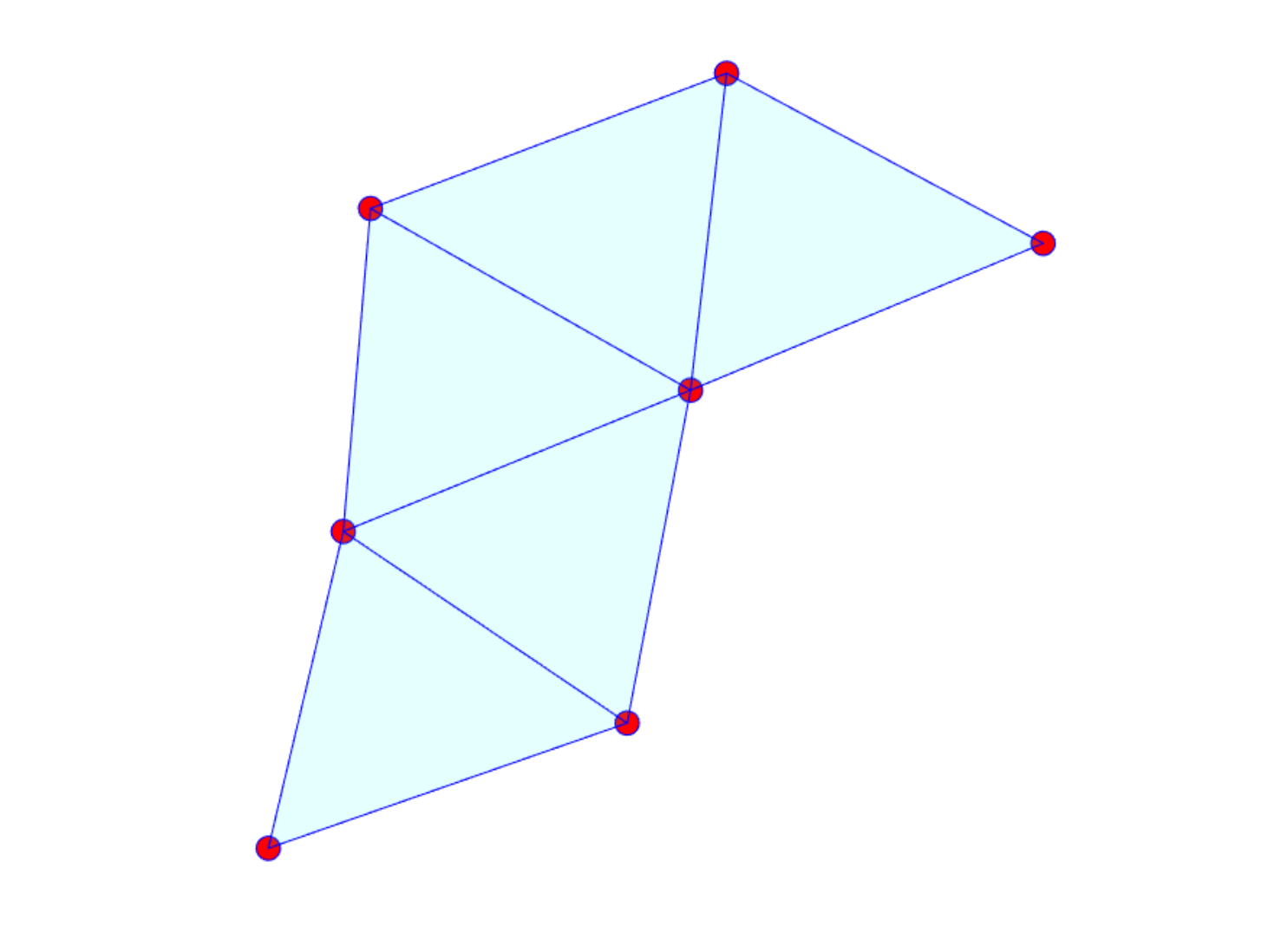}
\par\end{centering}
}

\caption{Illustration of the simplicial complexes $S_{5}^{2}$ (a), $t^{2}(1,2,4)$
(b) and $P_{5}^{2}$ (c). See text for definitions and notation.}

\label{Families}
\end{figure}

\section{Adjacency Matrices in Simplicial Complexes}

The goal of this section is to define a general adjacency matrix for
a simplicial complex that allows us to define general centrality indices
for these mathematical objects. Based on the previous definitions
of lower and upper adjacency relations we define the corresponding
adjacency matrices here.
\begin{defn}
Let $i$ and $j$ be two $k$-simplices in a simplicial complex. Then,
the \textit{lower adjacency matrix} $A_{l}^{k}$ at the $k$-level
in the simplicial complex has entries defined by

\[
(A_{l}^{k})_{ij}=\twopartdef{1}{\sigma_{i}\smile\sigma_{j}}{0}{\sigma_{i}\nsmile\sigma_{j}\text{ or }i=j},
\]
where the subindex $l$ indicates lower adjacency.
\end{defn}
In a similar way we have the following.
\begin{defn}
Let $i$ and $j$ be two $k$-simplices in a simplicial complex. Then,
the \textit{upper adjacency matrix} $A_{u}^{k}$ at the $k$-level
in the simplicial complex has entries defined by

\[
(A_{u}^{k})_{ij}=\twopartdef{1}{\sigma_{i}\frown\sigma_{j}}{0}{\sigma_{i}\nfrown\sigma_{j}\text{ or \ensuremath{i=j}}},
\]
where the subindex $u$ indicates upper adjacency.
\end{defn}
If two distinct $k$-simplices $\sigma_{i}$ and $\sigma_{j}$ are
upper adjacent then there exists some $(k+1)$-simplex $\tau=\{v_{0},v_{1},\dots v_{k+1}\}$
such that $\sigma_{i}\subset\tau$ and $\sigma_{j}\subset\tau$. Without
loss of generality we have $\sigma_{i}=\{v_{1},v_{2},\dots v_{k+1}\}$
and $\sigma_{j}=\{v_{0},v_{2},v_{3},\dots v_{k+1}\}$ then $\|\sigma_{i}\cap\sigma_{j}\|=\|\{v_{1},v_{2},\dots v_{k+1}\}\cap\{v_{0},v_{2},v_{3},\dots v_{k+1}\}\|=k-1$.
This means that $\sigma_{i}$ and $\sigma_{j}$ are also lower adjacent.
An alternative proof of this can be found in \citep{goldberg2002combinatorial}.

The above two definitions for two $k$-simplices to be adjacent leads
us to the problem that there are now four possible notions we can
use to define a general adjacency matrix for simplicial complexes.
The four possibilities are $A_{l}^{k},A_{u}^{k},A_{l}^{k}-A_{u}^{k},A_{l}^{k}+A_{u}^{k}$.
Each of these possible definitions of adjacency have pros and cons
as we explain in the next paragraph.

Simply using the lower adjacency matrix $A_{l}^{k}$ does not isolate
the effects of $k$-simplices from higher order simplices. In particular,
for $1$-simplices the lower adjacency matrix $A_{l}^{1}$ simply
describes the line graph of the network. The line graph is a transformation
of the graph in which the nodes of the line graph are the edges of
the graph, and two nodes of the line graph are connected if the corresponding
edges in the graph are incident to a common node. On the other hand,
using the upper adjacency matrix $A_{u}^{k}$ would ignore the effects
of any $k$-simplices which are not faces of higher simplices, meaning
that there is potential for a lot of information to be missed. For
instance, there could be many $2$-simplices (triangles) in a network
but not necessarily so many $3$-simplices, then the upper adjacency
matrix does not identify any of them as adjacent to each other. It
is worth noting that the traditional adjacency matrix of a network
corresponds to $A_{u}^{0}$ although two $0$-simplices cannot be
lower adjacent. Using the sum of the two adjacency matrices, would
emphasize the effects of the higher simplices over the lower ones.
However, it would lead to an adjacency matrix which features $2$'s
where two simplices are upper adjacent. What we want is an adjacency
matrix which indicates when two simplices are adjacent or not. Thus
this would not be appropriate. This leaves us with the difference
of the two adjacency matrices $A_{l}^{k}-A_{u}^{k}$ as our notion
of general adjacency.
\begin{defn}
For $k\geq1$ we have that two $k$-simplices are considered adjacent
if they are both lower adjacent and not upper adjacent. For $k=0$
two simplices shall be adjacent if they are upper adjacent. We shall
denote two $k$-simplices, $t_{i},t_{j}$ to be adjacent in the way
defined here by $t_{i}\backsim t_{j}$.
\end{defn}
This definition allows us to remove most of the effects of higher
simplices being adjacent in the adjacency matrix at the lower simplex
levels. A consequence of this is that it allows us to analyze the
relationships between the centralities of simplices and their faces
which we are particularly interested in at the node level. Secondly,
this notion of adjacency lines up nicely with the extensively studied
higher order Laplacians of simplicial complexes \citep{muhammad2006control}.
An off-diagonal entry of the higher order Laplacian matrix is non
zero if and only if the corresponding off-diagonal entry of $A_{l}^{k}-A_{u}^{k}$
is non-zero. This is the definition that shall be used in the rest
of this work. Further information on the Hodge Laplacian matrices
can be found in \citep{muhammad2006control,tahbaz2010distributed,muhammad2007decentralized,maletic2012combinatorial,goldberg2002combinatorial}.
Then we have the following important definition of adjacency matrix
of the simplicial complex.
\begin{defn}
Let $i$ and $j$ be two $k$-simplices in a simplicial complex. Then,
for $k\geq1$ the\textit{ adjacency matrix} $A^{k}$ at the $k$-level
in the simplicial complex has entries defined by

\[
(A^{k})_{ij}=\twopartdef{1}{\sigma_{i}\smile\sigma_{j}\text{ and }\sigma_{i}\nfrown\sigma_{j}}{0}{i=j\text{ or }\sigma_{i}\nsmile\sigma_{j}\text{ or }\sigma_{i}\frown\sigma_{j}},
\]
for $k=0$ the adjacency matrix shall be given by the upper adjacency
matrix.
\end{defn}

\section{Simplicial Shortest Path Distance}

In this section we will extend the concept of shortest path distance
to the different levels of a simplicial complex. We start by extending
the concept of walks to simplicial complexes.
\begin{defn}
Let $k\geq1$. Then, a $s^{k}$-walk is a sequence of alternating
$k$-simplices and $(k-1)$-simplices $s_{1},e_{1},s_{2},e_{2},\dots,e_{r-1},s_{r}$
such that for each $i\in\{1,\dots,r-1\}$ $e_{i}$ is a face of both
$s_{i}$ and $s_{i+1}$, and $s_{i}$ and $s_{i+1}$ are not both
faces of the same $(k+1)$-simplex. For $k=0$ a walk on the $0$-simplices
is just a walk in the normal graph-theoretic sense.
\end{defn}
On the simplicial complex from Figure \ref{Simplicla complex1}, we
have that $\{1,3,4\},\{3,4\},\{3,4,5\},\{4,5\},$ $\{4,5,6\},\{4,5\},\{3,4,5\},\{3,4\},\{2,3,4\}$
is an $s^{2}$-walk. Meanwhile, $\{6,9\},\{6\},\{6,7\},\{6\},$ $\{5,6\},\{5\},\{3,5\},\{3\},\{2,3\}$
is an $s^{1}$-walk.
\begin{defn}
A $s^{k}$-shortest path between two $k$-simplices $s_{a},s_{b}$
is a $s^{k}$-walk, $s_{a},e_{1},s_{2},e_{2},$ $\dots,s_{n},e_{n},s_{b}$,
such that $n$ is minimized. The value $n$ is the $s^{k}$-shortest
path length between the two $k$-simplices $s_{a},s_{b}$. We denote
this $d(s_{a},s_{b})=n$.
\end{defn}
It can be easily seen that the simplicial shortest path length between
two $k$-simplices is a proper distance. By definition $d(s_{a},s_{b})\geq0$
for all $s_{a},s_{b}\in R_{k}$where $R_{k}$ is the set of $k$-simplices.
Clearly $d(s_{a},s_{b})=0\Longleftrightarrow s_{a}=s_{b}$. To prove
$d(s_{a},s_{b})=d(s_{b},s_{a})$ then assume $d(s_{a},s_{b})=n$ then
the $s^{k}$-shortest path from $s_{a}$ to $s_{b}$ is of the form
$s_{a},e_{1},s_{2},e_{2},$ $\dots,s_{n-1},e_{n-1},s_{n},e_{n},s_{b}$.
This means that there is a $s^{k}$-walk from $s_{b}$to $s_{a}$
of the form $s_{b},e_{n},$ $s_{n},e_{n-1},s_{n-1},$ $\dots,e_{2},s_{2}$
$,e_{1},s_{a}$. We can then relabel $e_{1}\rightarrow e_{n},s_{2}\rightarrow s_{n},e_{2}\rightarrow e_{n-1},s_{3}\rightarrow s_{n-1},\dots,e_{n}\rightarrow e_{1}$
and so on to give a $s^{k}$-walk from $s_{b}$to $s_{a}$ of the
form $s_{b},e_{1},s_{2},e_{2},$ $\dots,s_{n-1},e_{n-1},s_{n},e_{n},s_{a}$
thus $d(s_{b},s_{a})\leq n$. If there was a $s^{k}$-walk shorter
than this then there would also be a $s^{k}$-walk from $s_{a}$ to
$s_{b}$ which was shorter than the original walk by symmetric arguments
thus $d(s_{b},s_{a})=n$ and $d(s_{a},s_{b})=d(s_{b},s_{a})$. To
prove $d(s_{a},s_{c})\leq d(s_{a},s_{b})+d(s_{b},s_{c})$ let $d(s_{a},s_{b})=n$
and $d(s_{b},s_{c})=m$ then there is a $s^{k}$-walk from $s_{a}$
to $s_{b}$ of the form $s_{a},e_{1},s_{2},e_{2},$ $\dots,s_{n-1},e_{n-1},s_{n},e_{n},s_{b}$
and $s^{k}$-walk from $s_{b}$ to $s_{c}$ of the form $s_{b},e_{1},s_{2},e_{2},$
$\dots,s_{m-1},e_{m-1},s_{m},e_{m},s_{c}$ we can combine these and
relabel the simplices in the second walk by the rules $s_{b}\rightarrow s_{n+1,}e_{i}\rightarrow e_{n+i},s_{i}\rightarrow s_{n+i}$
to form a $s^{k}$-walk from $s_{a}$ to $s_{c}$ of the form $s_{a},e_{1},s_{2},e_{2},$
$\dots,s_{n-1},e_{n-1},s_{n},e_{n},s_{n+1}$ $e_{n+1},s_{n+2},e_{n+2},$
$\dots,s_{n+m-1},e_{n+m-1},s_{n+m},e_{n+m},s_{c}$. This implies that
$d(s_{a},s_{c})\leq n+m=d(s_{a},s_{b})+d(s_{b},s_{c})$. For instance,
on the simplicial complex from Figure \ref{Simplicla complex1}, we
have that $\{1,3,4\},\{3,4\},$ $\{3,4,5\},\{3,4\},$ $\{2,3,4\}$
is a $s^{2}$-shortest path from $\{1,3,4\}$ to $\{2,3,4\}$ and
we have $d(\{1,$ $3,4\},$ $\{2,3,4\})=2$. Meanwhile, $\{2,4\},\{4\},\{4,6\},\{6\},\{6,7\}$
is a $s^{1}$-shortest path between $\{2,4\}$ and $\{6,7\}$ and
we have $d(\{2,4\},\{6,7\})=2$.
\begin{defn}
A simplicial complex is $s^{k}$-connected if and only if there does
not exist a pair of $k$-simplices $s_{a},s_{b}\in R_{k}$, where
$R_{k}$ is the set of $k$-simplices, such that $d(s_{a},s_{b})=\infty$.
\end{defn}
Note that a simplicial complex being $s^{k}$-connected does not mean
that it is $s^{k+1}$-connected or $s^{k-1}$-connected. The simplicial
complex in Figure \ref{Simplicla complex1} is $s^{0}$-connected
but not $s^{1}$-connected because $\{1,2\}$ and $\{7,8\}$ are not
adjacent to any of the other $1$-simplices. Many of the real world
networks we will introduce in a later section are $s^{1}$-connected
but not $s^{2}$-connected. In addition, a simplicial complex from
the family $S_{l}^{k}$ is $s^{k}$-connected but it is not $s^{k-1}$-connected.
The central $(k-1)$-simplex is upper adjacent to every other $(k-1)$-simplex
and hence is not adjacent to any of them.
\begin{defn}
An $s^{k}$-connected component of a simplicial complex is a subset
$S_{k}$ of the $k$-simplices $R_{k}$ such that for any two $k$-simplices
$s_{a},s_{b}\in S_{k}$ we have $d(s_{a},s_{b})<\infty$ and for any
$s\in S_{k}$ and $r\in R_{k}\setminus S_{k}$ we have that $d(s,r)=\infty$.

The $s^{k}$-eccentricity $\epsilon(t)$ of a $k$-simplex $s$ is
the largest $s^{k}$-shortest path distance between $s$ and any other
$k$-simplex. The $s^{k}$-diameter $D$ of a simplicial complex is
the maximum $s^{k}$-eccentricity of any simplex in the network $D=\max_{s\in R}\epsilon(s)$
where $R_{k}$ is the set of $k$-simplices. As an example, in the
simplicial complex $t^{2}(1,2,4)$, depicted in Figure \ref{Families},
the central $2$-simplex has $s^{2}$-eccentricity $1$ because it
is adjacent to all the other $2$-simplices in the complex. However
all the peripheral $2$-simplices have a $s^{2}$-eccentricity of
$2$ because the shortest path form a peripheral $2$-simplex on one
arm to a peripheral $2$-simplex on another is through the central
$2$-simplex for a shortest path of length $2$. This means that $t^{2}(1,2,4)$
has $s^{k}$-diameter $2$.
\end{defn}
Given a notion of shortest path distance we are now equipped to define
the average simplicial shortest path distance. The $s^{k}$-average
simplicial shortest path length is the average $s^{k}$-shortest path
distance for all possible $k$-simplices in the network
\begin{equation}
l_{k}=\frac{2\sum_{i<j}d_{k}(s_{i},s_{j})}{\|R_{k}\|(\|R_{k}\|-1)},
\end{equation}
where $R_{k}$ is the set of $k$-simplices in the network and $d_{k}(s_{i},s_{j})$
is the $s^{k}$-shortest path distance between $s_{i}$ and $s_{j}$.
Note for this measure to make any sense the simplicial complex needs
to be $s^{k}$-connected. If the simplicial complex is not $s^{k}$-connected
then we can analyze each $s^{k}$-connected component separately.
We will now prove bounds on the $s^{k}$-average path length. If we
assume that there are at least two $k$-simplices in the simplicial
complex. For $l_{k}$ to be less than $1$ there would need to be
two $k$-simplices, $s_{i},s_{j}$ such that $d(s_{i},s_{j})<1$ this
would imply $d(s_{i},s_{j})=0$ and hence $s_{i}=s_{j}$ by the properties
of a metric. The lower bound $l_{k}=1$ is achieved by a simplicial
complex of the form $S_{r}^{k}$. This is easy to check. A simplicial
complex of the form $S_{r}^{k}$ consists of a $(k-1)$-simplex $\{1,2,\dots,k\}$
and some $k$-simplices of the form $\{1,2,\dots,k,i\}$, where $i>k$,
in addition to all subsimplices necessary by the closure axiom. Hence,
all $k$-simplices are lower adjacent to each other by the $(k-1)$-simplex
$\{1,2,\dots,k\}$ and they are not upper adjacent to each other because
there are no $(k+1)$-simplices. Thus, every $k$-simplex is adjacent
to every other $k$-simplex and the $s^{k}$-shortest path distance
between any two $k$-simplices is $1$. Hence, the $s^{k}$-average
path length is $1$, which implies that the lower bound of $l_{k}$
is $1$.

A general upper bound of $l_{k}$ is hard to establish due to of the
dependence on the number of simplices, $\|R_{k}\|$. However, if we
fix both $k$ and $\|R_{k}\|$ then we can prove the following result.
\begin{lem}
Let $\|R_{k}\|$ be the number of $k$-simplices. Then, the upper
bound of $l_{k}$ is

\begin{equation}
\frac{\frac{(\|R_{k}\|-1)\|R_{k}\|(\|R_{k}\|+1)}{3}}{\|R_{k}\|(\|R_{k}\|-1)}=\frac{\|R_{k}\|+1}{3}.
\end{equation}
\end{lem}
\begin{proof}
Assume that the simplicial complex is $s^{k}$-connected and that
$\|R_{k}\|\geq2$. If $\|R_{k}\|=2$ then $\sum_{i<j}d_{k}(s_{i},s_{j})=1$,
the simplicial complex is $s^{k}$-connected and there are only 2
$k$-simplices hence they must be adjacent. Thus $l_{k}=\frac{2\sum_{i<j}d_{k}(s_{i},s_{j})}{\|R_{k}\|(\|R_{k}\|-1)}=1$.
In addition $\frac{\|R_{k}\|+1}{3}=1$. Hence the lemma holds for
$\|R_{k}\|=2$. Assume that the Lemma holds for $\|R_{k}\|\leq n$.
Let $\|R_{k}\|=n+1$ then to maximize $l_{k}$ we need to maximize
$\sum_{i<j}d_{k}(s_{i},s_{j})$. Pick a $k$-simplex $s_{1}$. First,
we will maximize $\sum_{j}d_{k}(s_{1},s_{j})$. For $d_{k}(s_{1},s_{j})=y$
for some $s_{j}\in R_{k}$, first it must be the case that $d_{k}(s_{1},s_{m})=y-1$
for some $s_{m}\in R_{k}$ such that $s_{m}\backsim s_{j}$. This
means that the largest possible value of $d_{k}(s_{1},s_{j})$ for
some $s_{j}\in R_{k}$ is $\|R_{k}\|-1=n$. This gives $\max\sum_{j}d_{k}(s_{1},s_{j})=(\|R_{k}\|-1)+(\|R_{k}\|-2)+\dots+1=T_{\|R_{k}\|-1}=T_{n}$
where $T_{z}$ represents the $z$th triangle number. Now this implies
that there is only one $k$-simplex, $s_{a}\in R_{k}$ such that $d_{k}(s_{a},s_{1})=1$.
This means that $s_{1}$ is adjacent to precisely one other $k$-simplex,
namely $s_{a}$. Because $s_{1}$ is adjacent to only one other simplex,
$s_{1}$ can be removed without affecting the $s^{k}$-shortest path
distances between any other $k$-simplices. We now have a simplicial
complex such that $\|R_{k}\|=n$. We know that the upper bound of
the $s^{k}$-average path distance for this smaller simplicial complex
is $\frac{\|R_{k}\|+1}{3}=\frac{n+1}{3}$ by assumption where $\frac{(n-1)n(n+1)}{6}$
is the contribution given by $\sum_{i<j}d_{k}(s_{i},s_{j})$. We also
know that the largest number we can add to the sum $\sum_{i<j}d_{k}(s_{i},s_{j})$
by the addition of a $k$-simplex is given by $T_{n}=\frac{n(n+1)}{2}$.
Thus $\max(\sum_{i<j}d_{k}(t_{i},t_{j}))=\frac{(n-1)n(n+1)}{6}+\frac{n(n+1)}{2}=\frac{n(n+1)(n+2)}{6}$.
This means that the upper bound of $l_{k}$ is $\frac{\frac{(\|R_{k}\|-1)\|R_{k}\|(\|R_{k}\|+1)}{3}}{(\|R_{k}\|-1)\|R_{k}\|}=\frac{\|R_{k}\|+1}{3}$.
Clearly as $\|R_{k}\|\rightarrow\infty$, $l_{k}\rightarrow\infty$
and so there is no upper bound for $l_{k}$. It should be fairly clear
that a simplicial complex of the form $P_{r}^{k}$ will achieve this
bound.
\end{proof}

\section{Simplicial Centralities}

\subsection{Centralities based on simplicial shortest-path}

We are now in a position to generalize some centrality notions for
simplices which are based on the simplicial shortest path distance.
The simplest of all centrality measures is the degree. In the case
of the simplicial complexes we have three levels of degrees, which
we will designate as $\delta_{k}\left(i\right)$, where $k=0,1,2$
is the level of the simplex, i.e., nodes, edges and triangles, respectively,
and $i$ is the corresponding simplex. The degree of a $k$-simplex
$s$ is the number of other $k$-simplices to which $s$ is adjacent.
If $p\left(\delta_{k}\right)$ is the probability of finding a $k$-simplex
of degree $\delta_{k}$ in a simplicial complex and $P\left(\delta_{k}\right)$
is the probability of finding a $k$-simplex of degree larger or equal
than $\delta_{k}$ in the simplicial complex, then the degree distribution
of the $k$-simplices is the probability distribution of the degrees
of the $k$-simplices across the whole of the simplicial complex.

Closeness centrality is a concept first introduced by Bavelas \citep{bavelas1950communication}
to capture the idea of how close\textemdash in terms of shortest path
distance\textemdash two nodes are in a network. Here we will generalize
this concept to simplicial complexes. The simplicial farness of a
$k$-simplex $F$ is the sum of its $s^{k}$-shortest path distances
to all other $k$-simplices, $\sum_{Y\neq F}d(Y,F)$. The simplicial
closeness is the reciprocal of simplicial farness. That is

\begin{equation}
C(F)=\frac{1}{\sum_{Y\neq F}d(Y,F)}
\end{equation}

Note that if the simplicial complex is not $s^{k}$-connected then
$\sum_{Y\neq F}d(Y,F)$ could be considered undefined or $\infty$
for all $k$-simplices in the simplicial complex. In this case we
can calculate simplicial harmonic closeness instead. This is a generalization
of a definition that can be found in \citep{rochat2009closeness}.
The simplicial harmonic closeness of a $k$-simplex $F$ is defined
as follows

\begin{equation}
H(F)=\sum_{Y\neq F}\frac{1}{d(Y,F)},
\end{equation}
where we treat $\frac{1}{\infty}=0$.

We would now like to establish some bounds on the simplicial closeness
centrality. However, there is an issue that needs to be considered
before bounds can be established. The issue is that the sum, $\sum_{Y\neq F}d(Y,F)$
depends on the number of simplices in the complex. If for all $Y\in R_{k}$
and $Y\neq F$, where $R_{k}$ is the set of $k$-simplices, $d(Y,F)=1$
then $C(F)=\frac{1}{\|R_{k}\|-1}$. Clearly this is the largest $C(F)$
can be for $\|R\|$ $k$-simplices. We can normalize this by multiplying
$C(F)$ by $(\|R\|-1)$ to give an upper bound of $C(F)=1$.
\begin{lem}
The upper bound of the normalized simplicial closeness centrality
can be attained by all simplices in a simplicial complex of the form
$S_{l}^{k}$.
\end{lem}
\begin{proof}
In a simplicial complex of the form $S_{l}^{k}$ we have $l=\|R_{k}\|$
$k$-simplices which are all adjacent to each other. Thus if we select
a particular $k$-simplex $s_{i}\in R_{k}$ we have that $d(s_{i},s_{j})=1$
for all $s_{j}\in R_{k}$ such that $s_{i}\neq s_{j}$. This gives
$\sum_{s_{i}\neq s_{j}}d(s_{i},s_{j})=\|R_{k}\|-1$. Hence $C(s_{i})=\frac{\|R_{k}\|-1}{\|R_{k}\|-1}=1$.
\end{proof}
We now prove a lower bound for the normalized simplicial closeness
centrality
\begin{lem}
Let us consider a $s^{k}$-connected simplicial complex with $\|R\|$
$k$-simplices. Then, the lower bound for the normalized simplicial
closeness centrality of a $k$-simplex is $0$, and this is attained
asymptotically when $\|R\|\rightarrow\infty$.
\end{lem}
\begin{proof}
Assume the simplicial complex is $s^{k}$-connected. We are trying
to minimize $\frac{1}{\sum_{Y\neq F}d(Y,F)}$ and hence trying to
maximize $\sum_{Y\neq F}d(Y,F)$. Take a $k$-simplex $F$ in a simplicial
complex $X$ which has $\|R\|$ $k$-simplices. Firstly, because $X$
is $s^{k}$-connected there exists a $s^{k}$-walk between $F$ and
every other $k$-simplex in the simplicial complex. The farthest distance
possible between $F$ and another simplex $s_{\|R\|}$ is $\|R\|-1$.
This means that the $s^{k}$-shortest path between these two simplices
looks like $F,e_{1},$ $s_{2},e_{2},$ $s_{3},\dots,$ $s_{\|R\|-1},$
$e_{\|R\|-1},s_{\|R\|}$. The shortest path between $F$ and any other
$k$-simplex $t_{n}$ must be the path $F,e_{1},s_{2},e_{2},s_{3},\dots,e_{n-1},s_{n}$,
i.e. the shortest path from $F$ to $s_{\|R\|}$ but cut off at simplex
$s_{n}$. If there was a shorter path from $F$ to $s_{n}$ then you
could replace this part of the path from $F$ to $s_{\|R\|}$ with
said shorter path from $F$ to $s_{n}$ and have a shorter path from
$F$ to $s_{\|R\|}$. Hence, the $s^{k}$-shortest path distance from
$F$ to any other $k$-simplex, $s_{n}$ is $n-1$. Note that for
$F$ to be at distance $r$ from a $k$-simplex $s_{r}$, $F$ must
first be at distance $r-1$ from a simplex adjacent to $s_{r}$. Hence,
the maximum value for $\sum_{Y\neq F}d(Y,F)=1+2+3+\dots+(\|R\|-1)=\frac{(\|R\|-1)\|R\|}{2}$.
This gives a lower bound on $C(F)$ of $\frac{2}{(\|R\|-1)\|R\|}$
which after normalization by multiplication by $(\|R\|-1)$ gives
a lower bound on $C(F)$ of $\frac{2}{\|R\|}$. This clearly tends
to $0$ as $\|R\|\rightarrow\infty$.
\end{proof}
The bound of $\frac{2}{\|R\|}$ for a given number of $k$-simplices
is achieved by the end simplex in a $t^{k}$-path.

If a simplicial complex is not $s^{k}$-connected then $C(F)=0$ or
it is considered undefined for all $k$-simplices $F\in R$. Thus,
the peripheral $2$-simplex which is only adjacent to the central
$2$-simplex in the complex $t^{2}(1,2,4)$ has simplicial closeness
given by $\frac{7}{13}$. We have $\sum_{Y\neq F}d(Y,F)=1+2+2+2+2+2+2=13$
where $Y$ is the given simplex and $F$ is a run through of the other
simplices. The $1$ is contributed by the shortest path from $Y$
to the central simplex while the $2$s are given by the shortest path
distances from $Y$ to the other peripheral simplices on the other
branches. While $\|R\|-1=7$ for the normalization.

To give an example from the simplicial complex in Figure \ref{Simplicla complex1}
we need to use the definition of simplicial closeness given in Definition
14. So to calculate the Simplicial closeness of $\{2,3,4\}$ we have
$H(\{2,3,4\})=\frac{1}{1}+\frac{1}{2}+\frac{1}{2}+\frac{1}{\infty}+\frac{1}{\infty}+\frac{1}{\infty}=2$.
This is because it is adjacent to $\{3,4,5\}$ and has shortest path
distance $2$ to both $\{1,3,4\}$ and $\{4,5,6\}$. There is no $s$-path
from $\{2,3,4\}$ to any of the other simplices.

The second centrality notion which is based on shortest paths that
we can generalize is the betweenness centrality. Betweenness centrality
was introduced by Freeman in 1977 in order to capture the notion of
how central a node in a network is in passing information through
other nodes. The following is a direct generalization of this definition
\citep{freeman1977set}. The simplicial betweenness of a $k$-simplex
$F$ is defined as follows

\begin{equation}
g(F)=\sum_{S\neq F\neq T}\frac{\sigma_{ST}(F)}{\sigma_{ST}}
\end{equation}
where $\sigma_{ST}$ is the total number of shortest paths from $S$
to $T$ and $\sigma_{ST}(F)$ is the number of such paths that pass
through $F$, where $F,S,T\in R_{k}$.

The betweenness centrality of a $k$-simplex increases as the number
of pairs of other simplices increases. It is therefore sensible to
divide $g(F)$ by $\frac{(\|R\|-1)(\|R\|-2)}{2}$, the number of pairs
of $k$-simplices which are not the simplex $F$. This gives a value
for simplicial betweenness in the range $[0,1]$. The lower bound
of $0$ is attained by every $k$-simplex in a simplicial complex
of the form $S_{l}^{k}$. It is also attained by any simplex which
is adjacent to only one other simplex. The upper bound of $1$ can
be attained by the central $k$-simplex of a $t^{k}(x_{1},x_{2},\dots,x_{k+1})$
simplicial complex where $x_{i}\in\{0,1\}\forall i\in\{1,2,\dots,k+1\}$.

\subsection{Spectral simplicial centralities}

We now move to the concepts of centrality based on spectral properties
of the simplicial complexes. Historically, for networks the first
of these centralities was developed by Katz \citep{katz1953new}.
The Katz centrality index tries to capture the notion that a node
in a network is not only influenced by its nearest neighbors but in
a lower extension by any other node separated at a given distance
from it, in a way in which such influence decays with the separation
between the nodes. In this section we generalize these ideas to simplicial
complexes largely following the example of \citep{estrada2015first}.

To make this task easier we define an underlying network of simplices
at every level of a simplicial complex. For all $k$ the adjacency
matrix of the $k$-simplices of a simplicial complex also gives rise
to a network where each node corresponds to a simplicial complex and
there is an edge between two nodes if and only if their corresponding
$k$-simplices are adjacent in the simplicial complex. We call this
network the underlying network of simplices. This immediately gives
us some results from network theory.
\begin{lem}
Let $A_{k}$ be the adjacency matrix between $k$-simplices in a simplicial
complex. Then, $(A_{k})_{ij}^{m}$ gives the number of $s^{k}$-walks
of length $m$ between $k$-simplex, $i$ and $k$-simplex, $j$.
\end{lem}
\begin{proof}
Every walk on the underlying network of simplices for a given simplex
of size $k$, has a corresponding $s^{k}$-walk over the $k$-simplices.
We have that $A_{k}$ is also the adjacency matrix for the nodes in
the underlying network of simplices. Thus, powers of the adjacency
matrix can be used to give the numbers of walks of a given length
on the underlying network of simplices. In particular, $(A_{k})_{ij}^{m}=b$
means that there are $b$ walks of length $m$ between node $i$ and
node $j$ in the underlying network of simplices at the $k$-simplex
level. This precisely corresponds to $s^{k}$-walks of length $m$
between simplex $i$ and simplex $j$. Simplex $i$ and simplex $j$
are the simplices represented by node $i$ and node $j$ respectively
in the underlying network of simplices.
\end{proof}
Let $A_{k}$ be the adjacency matrix representing the adjacency between
$k$-simplices in a simplicial complex. The simplicial Katz centrality
index is given by

\begin{equation}
K_{k,i}=[(\alpha^{0}A_{k}^{0}+\alpha A_{k}+\alpha^{2}A_{k}^{2}+\dots+\alpha^{m}A_{k}^{m}+\dots)\mathbf{(}e)]_{i}=\big[\sum_{m=0}^{\infty}(\alpha^{m}A_{k}^{m})\mathbf{e}\big]_{i},
\end{equation}
where $0<\alpha<\frac{1}{\lambda_{1}(A_{k})}$. The simplicial Katz
centrality is essentially the network-theoretic Katz centrality applied
to the underlying network of simplices. This means that as proved
in \citep{estrada2015first} the series, $(\alpha^{0}A_{k}^{0}+\alpha A_{k}+\alpha^{2}A_{k}^{2}+\dots+\alpha^{m}A_{k}^{m}+\dots)$
converges when $\alpha\le\rho(A_{k})$, where $\rho(A_{k})$ is the
spectral radius of $A_{k}$. This means that $K_{i}=[(I-\alpha A_{k})^{-1}\mathbf{e}]_{i}$.
We also have from \citep{estrada2015first} the representation of
the Katz centrality in terms the eigenvalues and eigenvectors of $A_{k}$.
This representation gives $K_{i}=\sum_{l}\sum_{j}\psi_{k,j}(i)\psi_{k,j}(l)\frac{1}{1-\alpha\lambda_{k,j}}$.
Where $\psi_{k,j}(i)$ and $\psi_{k,j}(l)$ are the $i$th and $l$th
entries of the $j$th eigenvector of $A_{k}$, respectively and $\lambda_{k,j}$
is the $j$th eigenvalue of $A_{k}$.

We can now use the simplicial Katz centrality to define the simplicial
eigenvector centrality. The following adjustment of the Katz centrality
appears in \citep{estrada2015first} and can also be applied to the
simplicial Katz centrality.

\begin{equation}
\begin{split} & \vec{v}_{k}=\big(\sum_{m=1}^{\infty}\alpha^{m-1}A_{k}^{m}\big)\mathbf{e}=\big(\sum_{m=1}^{\infty}\alpha^{m-1}\sum_{j=1}^{n}\vec{\psi}_{k,j}\vec{\psi}_{k,j}^{T}\lambda_{k,j}^{m}\big)\mathbf{e}\\
 & =\big(\frac{1}{\alpha}\sum_{j=1}^{n}\sum_{m=1}^{\infty}(\alpha\lambda_{k,j})^{m}\vec{\psi}_{k,j}\vec{\psi}_{k,j}^{T}\big)\mathbf{e}=\big(\frac{1}{\alpha}\sum_{j=1}^{n}\frac{1}{1-\alpha\lambda_{k,j}}\vec{\psi}_{k,j}\vec{\psi}_{k,j}^{T}\big)\mathbf{e}.
\end{split}
\end{equation}

Again following the example of \citep{estrada2015first} allows $\alpha$
to approach the inverse of the largest eigenvalue of $A_{k}$ from
below $(\alpha\to\frac{1}{\lambda_{1}}^{-})$. This gives

\begin{equation}
\begin{split} & \lim_{\alpha\to\frac{1}{\lambda_{k,1}}^{-}}(1-\alpha\lambda_{k,1})\vec{v}_{k}=\lim_{\alpha\to\frac{1}{\lambda_{k,1}}^{-}}\big(\frac{1}{\alpha}\sum_{j=1}^{n}\frac{(1-\alpha\lambda_{k,1})\mathbf{(}v)}{1-\alpha\lambda_{k,j}}\vec{\psi}_{k,j}\vec{\psi}_{k,j}^{T}\big)\mathbf{e}\\
 & =\big(\lambda_{k,1}\sum_{i=1}^{n}\psi_{k,1}(i)\big)\vec{\psi}_{k,1}=\gamma\vec{\psi}_{1,j}
\end{split}
\end{equation}

Therefore the eigenvector associated with the largest eigenvalue of
$A_{k}$ could also be said to be a centrality measure. This leads
us to the following definition. The simplicial eigenvector centrality
of the $i$th $k$-simplex in a simplicial complex is given by the
$i$th component of the principal eigenvector of $A_{k}$, $\psi_{k,1}(i).$

In a similar way as in the previous section we make a generalization
of the exponential of the adjacency matrix of $k$-simplices which
relies on results from the paper \citep{estrada2005subgraph} The
following power series of the adjacency matrix of $k$-simplices $A_{k}$
in a simplicial complex converges to the corresponding matrix exponential
\begin{center}
\begin{equation}
\sum_{l=0}^{\infty}\dfrac{A_{k}^{l}}{l!}=\exp\left(A_{k}\right).
\end{equation}
\par\end{center}

Obviously,
\begin{center}
\begin{equation}
\exp\left(A_{k}\right)=\sum_{l=0}^{\infty}\dfrac{A_{k}^{2l}}{\left(2l\right)!}+\sum_{l=0}^{\infty}\dfrac{A_{k}^{2l+1}}{\left(2l+1\right)!}=\cosh\left(A_{k}\right)+\sinh\left(A_{k}\right),
\end{equation}
\par\end{center}

where the first term accounts for the weighted sum of even-length
walks and the second one accounts for odd-length walks in the simplicial
complex.

We can now define a centrality measure analogous to subgraph centrality
for simplicial complexes. Subgraph centrality was introduced for networks
by Estrada and Rodríguez-Velázquez \citep{estrada2005subgraph} to
capture the participation of a node in a network in all subgraphs
in the network, giving more weight to the smaller than to the larger
ones. This is a direct generalization made possible by the adjacency
matrices at the different levels of the simplicial complex. Then,
the simplicial subgraph centrality of a $k$-simplex, $i$, is given
by $(\exp^{A_{k}})_{ii}$. For the simplicial complex in Figure \ref{Simplicla complex1}
we have that the simplicial subgraph centrality of the $1$-simplex
$\{1,4\}$ is 2.714 while the simplicial communicability between $\{1,4\}$
and $\{6,9\}$ is 2.0363. Note that any bounds on subgraph centrality
or simplicial communicability for networks still hold due to the underlying
network of simplices.
\begin{lem}
Let $k\geq1$ and let us consider an $s^{k}$-connected simplicial
complexes which contain a fixed number $b$ of $k$-simplices. The
upper bound of the simplicial subgraph centrality is attained by every
simplex in a simplicial complex $S_{b}^{k}$ and the lower bound is
attained by the two end simplices in a simplicial complex of the form
$P_{b}^{k}$.
\end{lem}
\begin{proof}
Fix a number of $k$-simplices to $b$. It is known that for $b$
nodes the upper bound of subgraph centrality in networks is attained
by every node in the complete graph $K_{b}$ \citep{estrada2005subgraph}.
The subgraph centrality for the underlying network of simplices at
the $k$-simplex level is the same as the simplicial subgraph centrality
for $k$-simplices. Thus, to find the upper bound of the simplicial
subgraph centrality we need to find a simplicial complex whose underlying
network of simplices is a complete graph. The simplicial complex $S_{b}^{k}$
satisfies this criterion. Similarly, the lower bound of subgraph centrality
in networks is attained by the two end simplices in a path graph of
length $b$. Thus to find the lower bound of the simplicial subgraph
centrality you need to find a simplicial complex whose underlying
network of simplices is a path graph. The simplicial complex $P_{b}^{k}$
satisfies this criterion.
\end{proof}

\section{Analysis of Protein Interaction Networks}

Here we study 10 protein-protein interaction (PPI) networks. In these
networks nodes represent proteins and undirected links represent the
interaction between two proteins determined experimentally. The networks
studied correspond to the following organisms: \textit{D. melanogaster}
(fruit fly) \citet{giot2003protein}, Kaposi sarcoma herpes virus
(KSHV) \citet{uetz2006herpesviral}, \textit{P. falsiparum} (malaria
parasite) \citet{lacount2005protein}, varicella zoster virus (VZV)
\citet{uetz2006herpesviral}, human \citet{rual2005towards}, \textit{S.
cereviciae} (yeast) \citet{bu2003topological}, \textit{A. fulgidus}
\citet{motz2002elucidation}, \textit{H. pylori} (\citet{lin2004hp,rain2001protein}),
\textit{E. coli \citet{butland2005interaction}} and \textit{B.subtilus}
\citet{noirot2004protein}. We study only the largest (main) connected
component of each of these networks, which range from 50 to 3039 proteins.
We then transformed these networks into their clique simplicial complexes
consisting of edges and of triangles, respectively. The number of
simplices and interactions at the nodes, edges and triangle level
are given in Table 1. Notice that the number of simplices at the edges
level is the same as the number of interactions at the nodes level.

\begin{table}
\begin{centering}
\begin{tabular}{>{\centering}p{3cm}>{\centering}m{1.6cm}>{\centering}m{1.6cm}>{\centering}m{1.6cm}>{\centering}m{1.6cm}>{\centering}m{1.6cm}}
\hline
 & \multicolumn{2}{c}{nodes} & edges & \multicolumn{2}{c}{triangles}\tabularnewline
\hline
species & simplices & interact. & interact. & simplices & interact.\tabularnewline
\hline
\textit{A. fulgidus} & 32 & 37 & 101 & 1 & 0\tabularnewline
\hline
KSHV & 50 & 114 & 606 & 34 & 82\tabularnewline
\hline
VZV & 53 & 148 & 1156 & 104 & 343\tabularnewline
\hline
\textit{B.subtilus} & 84 & 98 & 463 & 4 & 1\tabularnewline
\hline
\textit{P. falsiparum} & 229 & 604 & 4599 & 201 & 401\tabularnewline
\hline
\textit{E. coli} & 230 & 695 & 7803 & 478 & 2425\tabularnewline
\hline
\textit{H. pylory} & 710 & 1396 & 14736 & 76 & 79\tabularnewline
\hline
\textit{S. cereviciae} & 2224 & 6609 & 99882 & 3530 & 15004\tabularnewline
\hline
human & 2783 & 6007 & 85617 & 1047 & 2170\tabularnewline
\hline
\textit{D. melanogaster} & 3039 & 3687 & 11369 & 163 & 113\tabularnewline
\hline
\end{tabular}\caption{Number of simplices and their interactions at the nodes, edges and
triangles levels for the 10 PPI networks studied. Notice that the
number of simplices at the edge level is the same as the number of
interactions at the node level.}
\par\end{centering}
\label{data simplices}
\end{table}

\subsection{Degree distributions}

The study of node degree distribution has become one of the standard
tests considered for the structural analysis of networks. A network
with a broad degree distribution\textendash also know as fat-tailed
distribution\textendash is characterized by the presence of a few
hubs\textendash high degree node\textendash which keep the network
together. These hubs are important from the structural and functional
point of view in these networks. In the case of PPI networks hubs
are expected to play fundamental role in the cell and their knockout
is expected to produce a large cellular damage. This is the main hypothesis
of the centrality-lethality paradigm. A particular kind of fat-tailed
distribution, the power-law one, received a large deal of attention
in the literature. A power-law degree distribution is also know as
a scale-free distribution and it is indicative of some self-similarities
properties in the network. At the beginning of the XXI century a deluge
of papers finding scale-free distributions in almost every network
were published. Many of the existing PPI networks were characterized
as scale-free ones based on these findings. Later, more order has
being in place and some authors have found that almost none of the
PPI networks previously claimed to have scale-free structures were
so \citep{stumpf2005probability}. The main message of these experiences
is that most of PPI networks indeed display some kind of heavy-tailed
degree distributions, such as power-law, lognormal, Burr, logGamma,
Pareto, etc. However, as we will see here this is not necessarily
true when a large number of statistical distributions and goodness
of fit parameters are tested for the 10 PPI networks considered in
this work.

Here we consider the probability degree functions (PDF), $p\left(\delta_{k}\right)$
vs. $\delta_{k}$, for 10 PPI networks at the three different levels
studied in this work, i.e., nodes, edges, and triangles. For each
of the PDFs we fit the data to every of the following distributions:
Beta, Binomial, Birnbaum-Saunders, Burr, Exponential, Extreme Value,
Gamma, Generalized Extreme Value (GEV), Generalized Pareto (gen-Pareto),
Half-normal, Inverse Gaussian, Kernel, Logistic, Loglogistic, Lognormal,
Nakagami, Negative Binomial, Normal, Poisson, Rayleigh, Rician, Stable,
$t$ Location-Scale, and Weibull. The best fit was determined on the
basis of the following statistical parameters: Akaike information
criterion (AIC) \citep{konishi2008information,symonds2011brief} and
the Bayesian information criterion (BIC) \citep{konishi2008information}.
These indices are defined as follow:

\begin{equation}
AIC=2k-2\ln\left(\hat{L}\right),
\end{equation}

\begin{equation}
BIC=k\ln\left(n\right)-2\ln\left(\hat{L}\right),
\end{equation}
where $n$ is the number of data points, $k$ is the number of parameters
to be estimated and $\hat{L}=p\left(x|\hat{\theta},M\right)$ is the
maximized value of the likelihood function of the model $M$, where
$\hat{\theta}$ are the parameter values that maximize the likelihood
function and $x$ are the data points. For a series of models trying
to describe the same dataset, the smallest values of these three parameters
gives the best fit for the data. However, it is important to consider
the differences between the values of these parameters for the corresponding
models as we describe below.

We then first fit the dataset corresponding to the degrees of the
corresponding simplices in a PPI to all the studied distributions.
Then, we rank all the distributions in increasing order of their AIC.
We then compare the values of the first few distributions in the ranking
using $\varDelta AIC_{i}=\exp\left(\left(AIC_{min}-AIC_{i}\right)/2\right)$,
where $AIC_{min}$ is the AIC for the top distribution in the ranking.
If $\varDelta AIC_{i}<0.01$ we consider that the first distribution
in the ranking is significantly different from the second (and consequently
the rest) as to accept it as the most significant one. In those cases
where the differences in the AIC is not significant we also consider
the difference in the BIC values. In this case we apply the Kass-Raftery
criterion as follows:
\begin{center}
\begin{tabular}{|>{\centering}p{2cm}|>{\centering}p{3cm}|}
\hline
$\varDelta BIC_{i}$ & meaning\tabularnewline
\hline
\hline
0-2 & not significant\tabularnewline
\hline
2-6 & positive\tabularnewline
\hline
6-10 & strong\tabularnewline
\hline
\textgreater{}10 & very strong\tabularnewline
\hline
\end{tabular}
\par\end{center}

This means, for instance, that if the difference in the values of
BIC is not bigger than 2, this criterion is not able to distinguish
between the two distributions. If, however, it is between 6-10 there
is a strong criterion to consider the distribution with the smallest
BIC as the most significant one \citep{kass1995bayes}. In Table 2
we show the best distribution fitted for each of the datasets studied.

\begin{table}
\begin{centering}
\begin{tabular}{>{\centering}p{3cm}>{\centering}p{3.5cm}>{\centering}p{3.5cm}>{\centering}p{3.5cm}}
\hline
species & nodes & edges & triangles\tabularnewline
\hline
\textit{A. fulgidus} & NA & gen-Pareto & NA\tabularnewline
KSHV & NA & gen-Pareto & gen-Pareto{*}\tabularnewline
VZV & gen-Pareto{*} & gen-Pareto/gamma{*}{*} & NA\tabularnewline
\textit{B. subtilus} & gen-Pareto & NA & NA\tabularnewline
\textit{P. falsiparum} & NA & gamma & gen-Pareto\tabularnewline
\textit{E. coli} & gen-Pareto & GEV{*} & gen-Pareto\tabularnewline
\textit{H. pylory} & gen-Pareto & gamma & NA\tabularnewline
\textit{S. cereviciae} & gen-Pareto & GEV & gen-Pareto\tabularnewline
human & gen-Pareto & NA & NA\tabularnewline
\textit{D. melanogaster} & gen-Pareto & GEV & gen-Pareto\tabularnewline
\end{tabular}\caption{Degree distributions of the nodes, edges and triangles in the simplicial
complexes representing 10 PPI networks studied here (see text for
selection criteria). Not available (NA) distributions are reported
when the data was scarce for a statistically significant fit of the
distributions or the statistical criteria used were unable to decide
between two or more distributions. {*}BIC criterion indicates only
a strong differentiation with the second best distribution. {*}{*}BIC
criterion indicates only a positive differentiation with the second
best distribution (see Appendix). }
\par\end{centering}
\label{distributions}
\end{table}

The most interesting observation from the results shown in Table 2
is that all distributions obtained for the three levels of the simplicial
complexes of the 10 PPI networks studied are heavy-tailed distributions.
At the node level, the 7 distributions that were statistical significant\textendash for
the other three the statistical criteria used were not able to distinguish
between the first few distributions\textemdash correspond to a generalized
Pareto distribution, where the probability of finding nodes of a given
degree decays as a power-law of the corresponding degree (see Appendix).
At the edge level, the PPI networks display GEV, generalized Pareto,
and gamma distributions, all of which are heavy-tailed (see Appendix).
Finally, at the triangles level 5 PPI networks display generalized
Pareto distributions and for the others it was not possible to determine
the best distribution. These results indicate that at the three levels
studied here, nodes, edges and triangles, there are simplicial-hubs
which concentrate most of the connectivity of the simplicial complexes
at the corresponding level. The damage of these hubs is expected to
produce catastrophic consequences for the functionality of the cell.
On the other side of the coin, the existence of heavy-tailed distributions
guarantees that the corresponding simplicial complexes are more robust
at these levels to the random failure of simplices. These results
also point out to the necessity of using other types of characterization
of the degree heterogeneity for simplicial complexes by considering
not-statistical indices, which can be applied even for small datasets
and/or datasets with small variability in their degrees \citep{estrada2010degreehet}.
This is an ongoing project in our group which will be considered in
a separate work.

\subsection{Comparison of centralities at different levels}

Simplicial centrality measures are all designed to identify the ``most
important'' simplices in a simplicial complex at different levels
and according to certain topological feature of the complex, such
as nearest-neighbor connectivity (degrees), proximity of other simplices
(closeness) and participation of a simplex in small sub-complexes
with other simplices (subgraph centralities). Then, it is expected
that there is some correlation between the centralities inside each
level of analysis. That is, it is expected that node degree is somehow
correlated to node closeness or node subgraph centrality for a given
PPI. For instance, in the yeast PPI the node centralities (degree,
closeness and subgraph centrality) have an average rank correlation
coefficient $\left\langle r_{n,n}\right\rangle \approx0.828$, with
the hugest rank correlation coefficient $r$ being between the closeness
and the subgraph centralities ($r\approx0.924$). At the edges level
this average rank correlation is of $\left\langle r_{e,e}\right\rangle \approx0.827$
and at the triangle level it raises up to $\left\langle r_{t,t}\right\rangle \approx0.970$.

In contrast with what we expect, and observe, at the individual levels
of the simplicial complex, is what we should expect on the relations
between two different levels of the simplicial complex. That is, we
do not have any theoretical insight indicating whether the information
provided by the centralities at the node level is or is not correlated
to that provided at the edges or triangles ones. In this case, however,
it should be desirable that not so high rank correlation is observed
as a way to increase the amount of different structural information
encoded by the simplicial centralities. This is indeed what is observed
for the PPI simplicial complex of yeast. The average rank correlation
coefficient between the nodes and edges centralities is just $\left\langle r_{n,e}\right\rangle \approx0.609$,
and that between nodes and triangles centralities is $\left\langle r_{n,t}\right\rangle \approx0.587$.
Finally, the average rank correlation between the edges and triangles
levels is barely $\left\langle r_{e,t}\right\rangle \approx0.228$.
These lack of correlations between the inter-level centralities (see
Table \ref{Correlation Table}) clearly indicate that the top nodes
in the ranking at one simplicial level does not necessarily coincide
with that produced by the centralities at a different level.

\begin{table}
\begin{centering}
\begin{tabular}{|c|c|c|c|c|c|c|c|c|}
\hline
 & \multicolumn{2}{c|}{nodes} & \multicolumn{3}{c|}{edges} & \multicolumn{3}{c|}{triangles}\tabularnewline
\hline
 & SC & CC & DC & SC & CC & DC & SC & CC\tabularnewline
\hline
\hline
Node Degree & 0.7602 & 0.7984 & 0.3292 & 0.3856 & 0.3020 & 0.6586 & 0.6951 & 0.6812\tabularnewline
\hline
Node Subgraph &  & 0.9245 & 0.7376 & 0.6981 & 0.7471 & 0.5562 & 0.5780 & 0.5990\tabularnewline
\hline
Node Closeness &  &  & 0.7347 & 0.7710 & 0.7805 & 0.4795 & 0.5103 & 0.5260\tabularnewline
\hline
Edge Degree &  &  &  & 0.7617 & 0.9025 & 0.2744 & 0.2794 & 0.2928\tabularnewline
\hline
Edge Subgraph &  &  &  &  & 0.8180 & 0.2042 & 0.2207 & 0.2456\tabularnewline
\hline
Edge Closeness &  &  &  &  &  & 0.1558 & 0.1691 & 0.2076\tabularnewline
\hline
Triangle Degree &  &  &  &  &  &  & 0.9725 & 0.9772\tabularnewline
\hline
Triangle Subgraph &  &  &  &  &  &  &  & 0.9589\tabularnewline
\hline
\end{tabular}
\par\end{centering}
\caption{Spearman's rank correlation coefficients between the rankings of three
centralities of the 0,1 and 2-simplices in the yeast PPI.}

\label{Correlation Table}
\end{table}

It is also important to consider that none of the correlations are
negative. This implies that none of the centralities fundamentally
disagree with each other. It is not the case that a centrality at
one level is telling us that one set of nodes is not important and
another set of nodes is, while a centrality at a different level is
telling us the exact opposite. It is more likely that a centrality
at one level is telling us that one set of nodes is important while
a centrality at a different level is telling us the same thing but
the order of importance is shuffled between the two centralities.
This hypothesis is backed up when we consider the triangle and node
degrees. Of the 100 most central nodes according to these centralities
24 coincide. When we consider the top 300 this rises to 111 proteins
(37\%) and looking at the top 500 the two centralities identify 268
(53.6\%). This may explain the difference between the centralities
capabilities in the detection of essential proteins in the next section
when a small percentage of the top proteins are considered compared
to the similarity when a larger percentage is considered.

We then study the average rank correlation coefficient for all the
PPI networks considered in this work. In Table (\ref{average correlation})
we give the average Spearman rank correlation coefficients for all
the PPI networks studied. As can be seen the inter-level correlations
between the centralities considered is relatively high following our
expectations of different centralities identifying essentially the
same groups of proteins at each corresponding level. The highest correlations
are observed for the triangle level, which is mainly due to the high
correlation between the triangle degree and closeness centralities.
This high correlation could be a consequence of the fact that most
of the high degree triangles are clumped together forming complexes
of many other triangles (see next section for the case of yeast).
Then, these high-degree triangles are close to each other, giving
also a high triangle closeness. Finally, we also observe poor rank
correlation between the different pairs of levels considered for the
10 PPI networks analyzed. The slightly negative average obtained for
$\left\langle r_{e,t}\right\rangle $ in \textit{B. subtilus }can
be considered more like a lack of correlation than as a negative correlation
between the indices because in no case the Spearman correlation coefficient
is larger, in modular terms, than 0.1.

\begin{table}
\begin{centering}
\begin{tabular}{|c|c|c|c|c|c|c|}
\hline
 & $\left\langle r_{n,n}\right\rangle $ & $\left\langle r_{e,e}\right\rangle $ & $\left\langle r_{t,t}\right\rangle $ & $\left\langle r_{n,e}\right\rangle $ & $\left\langle r_{n,t}\right\rangle $ & $\left\langle r_{e,t}\right\rangle $\tabularnewline
\hline
\hline
\textit{A. fulgidus} & 0.822 & 0.844 & NA & 0.399 & NA & NA\tabularnewline
\hline
KSHV & 0.912 & 0.778 & 0.993 & 0.632 & 0.666 & 0.493\tabularnewline
\hline
VZV & 0.873 & 0.783 & 0.899 & 0.176 & 0.751 & 0.186\tabularnewline
\hline
\textit{B. subtilus} & 0.749 & 0.865 & 0.792 & 0.407 & 0.294 & -0.030\tabularnewline
\hline
\textit{P. falsiparum} & 0.842 & 0.826 & 0.957 & 0.608 & 0.655 & 0.403\tabularnewline
\hline
\textit{E. coli} & 0.842 & 0.915 & 0.959 & 0.714 & 0.699 & 0.483\tabularnewline
\hline
\textit{H. pylory} & 0.847 & 0.740 & 0.932 & 0.608 & 0.458 & 0.254\tabularnewline
\hline
\textit{S. cereviciae} & 0.828 & 0.827 & 0.970 & 0.609 & 0.587 & 0.228\tabularnewline
\hline
human & 0.732 & 0.818 & 0.929 & 0.641 & 0.505 & 0.235\tabularnewline
\hline
\textit{D. melanogaster} & 0.661 & 0.703 & 0.795 & 0.568 & 0.303 & 0.188\tabularnewline
\hline
\end{tabular}
\par\end{centering}
\caption{Intra- ($\left\langle r_{n,n}\right\rangle $, $\left\langle r_{e,e}\right\rangle $
and $\left\langle r_{t,t}\right\rangle $) and inter-level ($\left\langle r_{n,e}\right\rangle $,
$\left\langle r_{n,t}\right\rangle $ and $\left\langle r_{e,t}\right\rangle $)
average Spearman's rank correlation coefficients between the rankings
of three centralities of the 0, 1 and 2-simplices in the 10 PPI networks
studied. See text for notation and explanations.}

\label{average correlation}
\end{table}

\subsection{Identification of essential proteins}

\subsubsection{Methodology}

An essential protein is the one that if knocked out the cell dies.
Then, the identification of such proteins has become one of the main
paradigms of the study of centrality measures in PPIs \citep{seringhaus2006predicting,yu2007importance,wang2012identification,gustafson2006towards}.
The reasons for such interest are twofold. On the one hand, it is
important to have theoretical tools that allow the identification
of proteins that can be drug targets, think for instance in the identification
of essential proteins in a pathogenic microorganism. On the other
hand, it is one of the scarce examples in which centrality measures
can be validated against some experimental data. The methodology for
essential protein identification that we consider here is adapted
from the one developed by Estrada in 2006, and consists of the following
steps \citep{estrada2006virtual}. First, we transform the PPI network
into a clique simplicial complex to consider node, edge and triangle
centralities. Then, we calculate the corresponding centralities at
the three levels for each of the simplicies. We transform edge and
triangle centrality into information based on the nodes forming such
fragments by calculating the average centrality of all the edges and
triangles in which a given node is involved in, respectively. Using
these centralities we rank all the proteins in the PPI in decreasing
order of their centralities. We then count how many essential proteins
are in the top $x$\% of the ranking and report this number as the
percentage of essential proteins identified by the corresponding centrality
(see Fig. \ref{Protein identification}). An ideal index for essential
protein identification will be the one which rank all essential proteins
at the top of the ranking, such as if we want to select 100 essential
proteins we simply pick the top 100 proteins in that ranking. We always
compare the results of this process with the random selection of proteins.
That is, we rank randomly the proteins in the PPI and count the number
of essential ones in the top $x$\% of the ranking.

\begin{figure}
\begin{centering}
\includegraphics[width=0.75\textwidth]{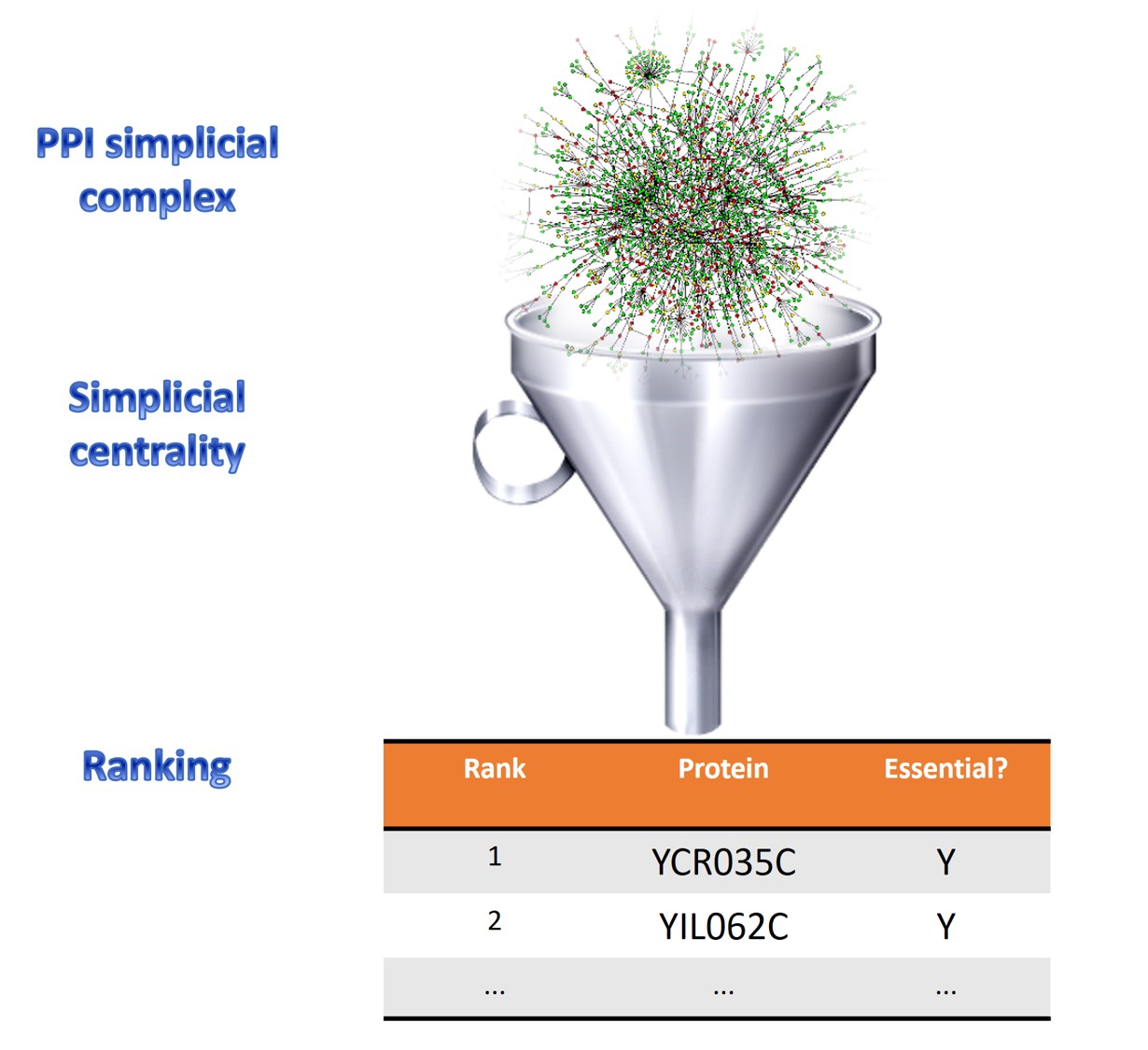}
\par\end{centering}
\caption{Schematic representation of the process of identification of essential
proteins using simplicial centralities in a PPI.}

\label{Protein identification}
\end{figure}

\subsubsection{Application to yeast PPI}

We now apply the methodology previously described to identify essential
proteins in the yeast PPI. There are several datasets of the interactions
of proteins in yeast. Here we use the data compiled by Bu et al. \citep{bu2003topological}.
The original data was obtained by von Mering et al. \citep{von2002comparative}
by assessing a total of 80,000 interactions among 5400 proteins reported
previously and assigning each interaction a confidence level. Bu et
al. \citep{bu2003topological} focused on 11,855 interactions between
2617 proteins with high and medium confidence in order to reduce the
interference of false positives, from which they reported a network
consisting on 2361 nodes and 6646 links (http://vlado.fmf.uni-lj.si/pub/networks/data/bio/Yeast/Yeast.htm).
This interaction map is considered here as a network in which proteins
are represented as the nodes and two nodes are linked by an edge if
the corresponding two proteins can be expected with high or medium
confidence of interacting. In this section we consider the node, edge
and triangle degree, closeness, and subgraph centralities as examples
of nearest-neighbor, shortest-path and spectral centralities.

The first interesting observation obtained from this analysis is that
the centralities based on the edge level of the simplicial complex
perform very badly in identifying the essential proteins. For instance,
for the top 1\% of proteins in the ranking, the node and triangle
centralities identify more than 45\% of essential protein (see detailed
analysis below), but the edge centralities do not identify more than
10\% of them (edge degree identifies 27\%). In general, neither of
the edge centrality is able to identify more than 35\% of essential
proteins at any percentage of top proteins selected. This result contrast
very much with the ones obtained by using node and triangle centralities.
For instance, for the closeness centrality at both node and triangle
level, the number of essential proteins identified is always larger
than 37\%. As can be seen in Fig. (\ref{centrality_essentiality}(a))
the triangle closeness centrality significantly outperforms the node
one for all the percentages of proteins considered. Triangle closeness
can identify up to 10\% more essential proteins than the same centrality
based on nodes, e.g., for 10\% and 15\% of top proteins. These differences
represent up to 40 essential proteins more identified by the triangle
centrality than the ones identified by the node one.

The largest percentages of essential proteins identified are obtained
by means of the subgraph centralities. In particular, for 1\% and
3\%, the triangle subgraph centrality outperforms the node one in
significant proportions. For 1\% of top proteins the triangle subgraph
centrality identifies 20\% more of essential proteins than its node
analogous, and for the 3\% it outperforms the node centrality in 14\%.
However, for top percentages of rankings over 5\%, the node and triangle
subgraph centrality do not show very significant differences in the
identification of essential proteins and they both identify around
50\% of the essential proteins existing.

\begin{figure}
\subfloat[]{\begin{centering}
\includegraphics[width=0.45\textwidth]{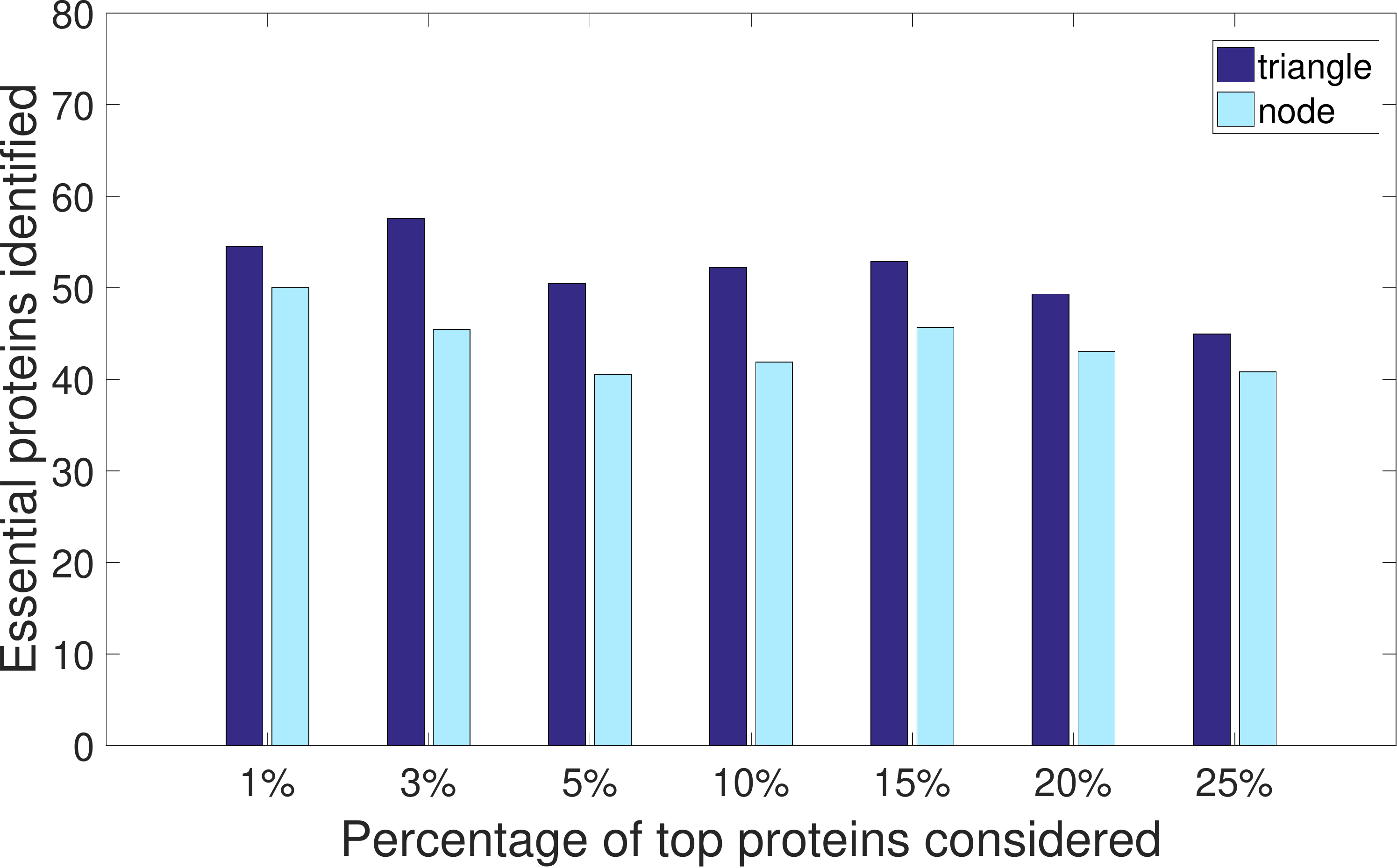}
\par\end{centering}
}\subfloat[]{\begin{centering}
\includegraphics[width=0.45\textwidth]{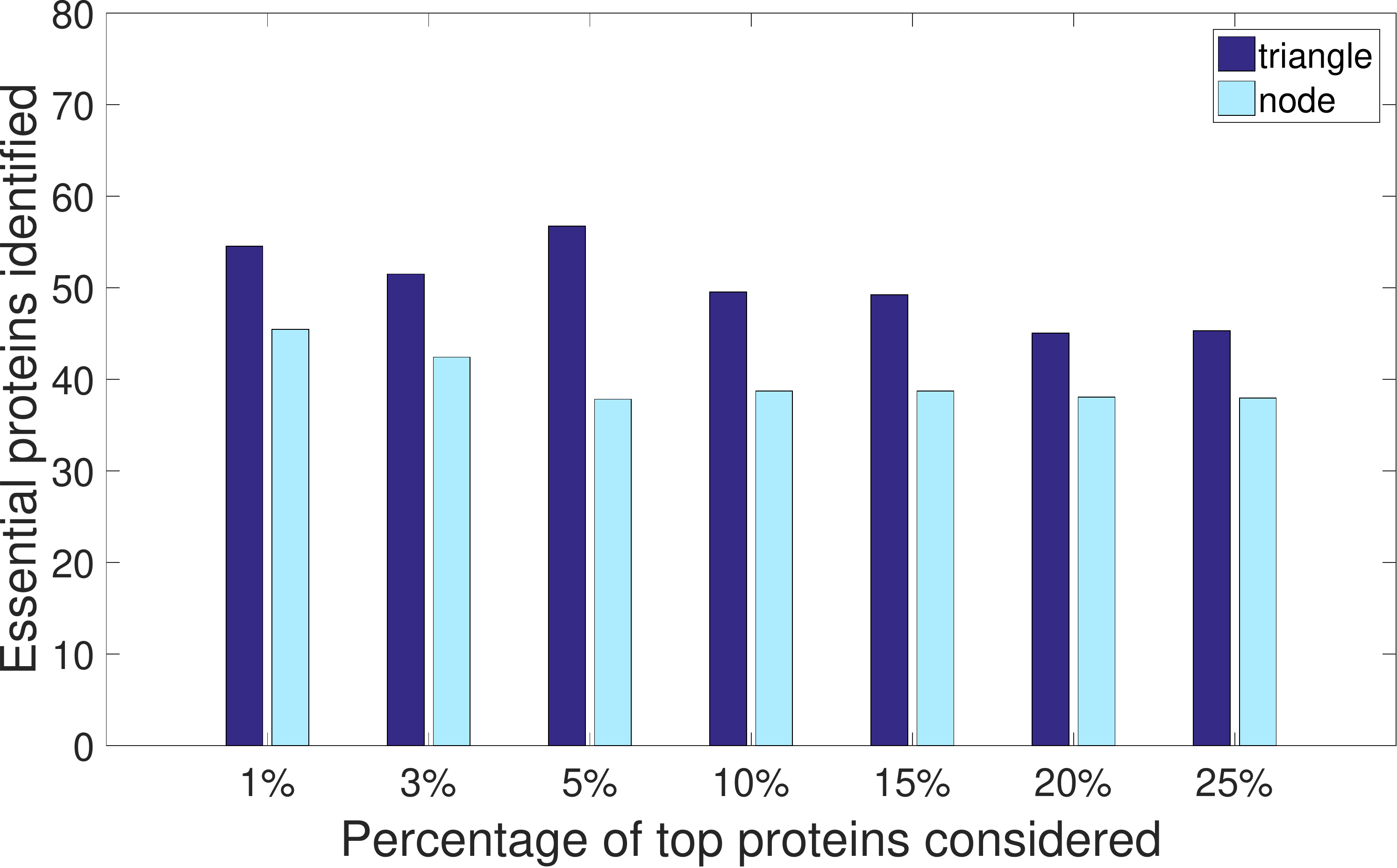}
\par\end{centering}
}

\subfloat[]{\begin{centering}
\includegraphics[width=0.45\textwidth]{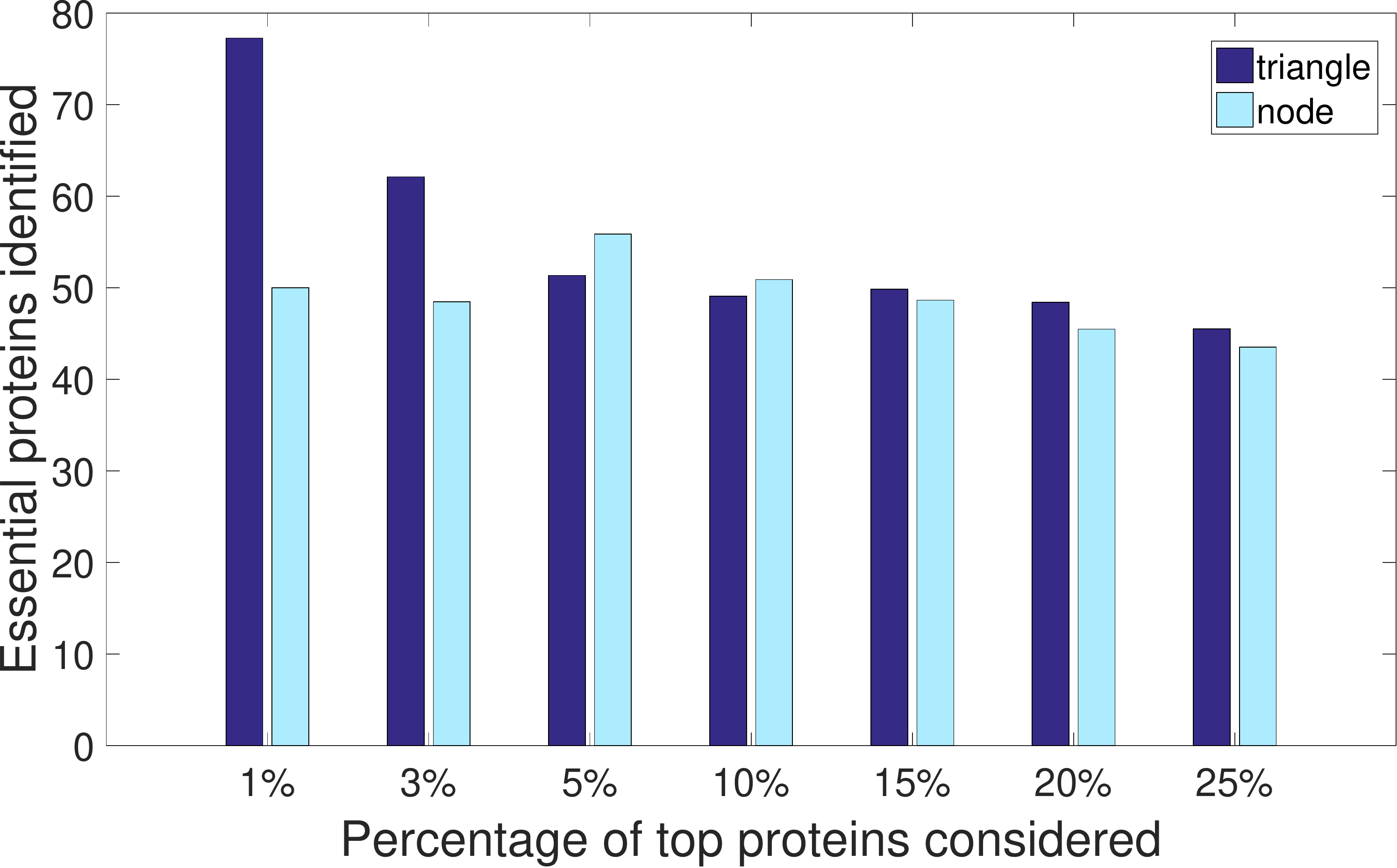}
\par\end{centering}
}\subfloat[]{\begin{centering}
\includegraphics[width=0.45\textwidth]{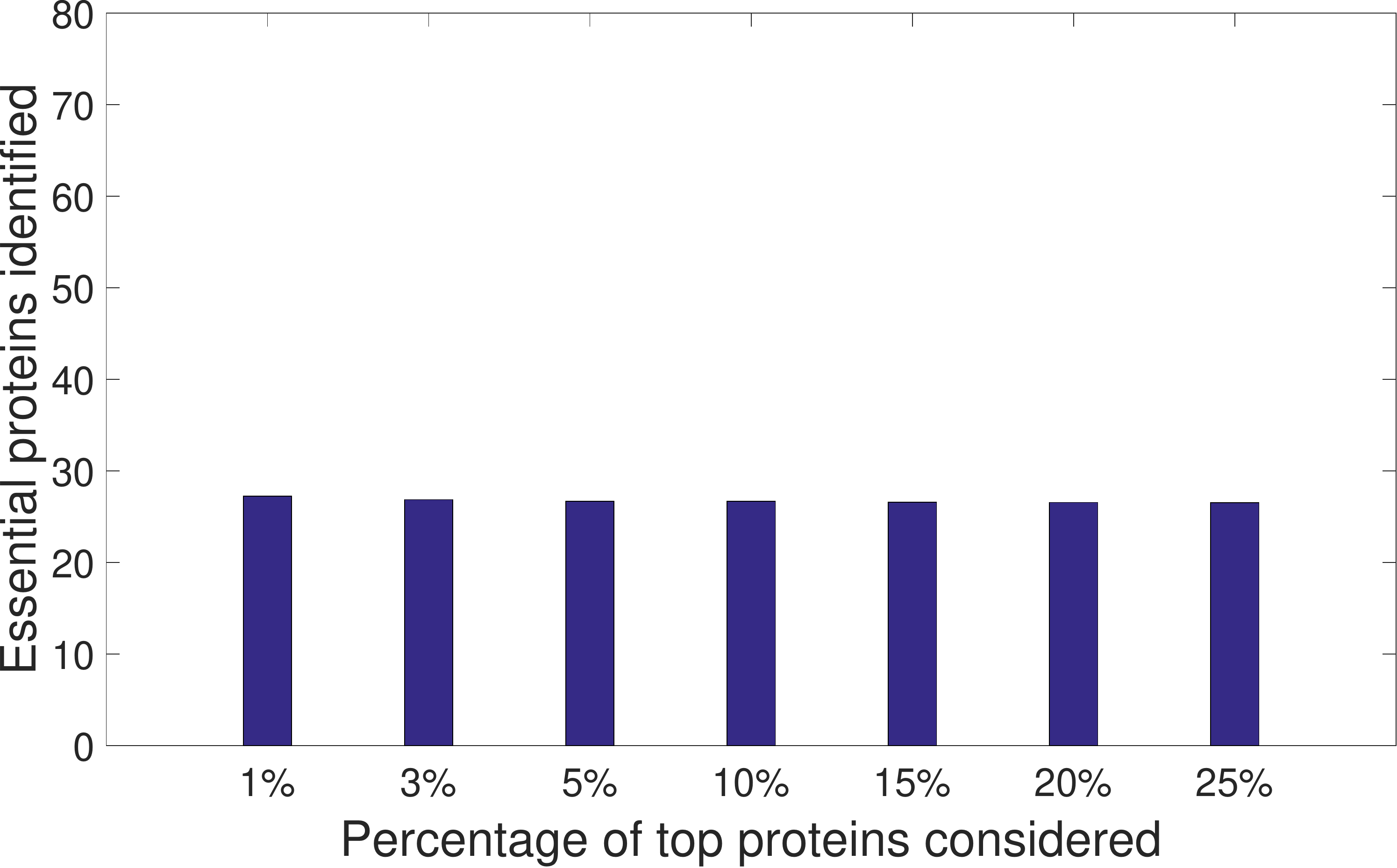}
\par\end{centering}
}

\caption{Percentage of essential proteins identified using simplicial degree
(a), closeness (b), simplicial subgraph centrality (c) and random
selection (d) based on nodes and of triangles of the simplicial complex
for yeast PPI. For the case of the random selection, as the edges
and triangle information is reduced to values for the nodes, it is
only needed the selection of essential proteins based on random ranking
of the nodes.}

\label{centrality_essentiality}
\end{figure}

The first observation which merits to be explained from a theoretical
point of view is the fact that edge centralities do not describe correctly
the essentiality of proteins in the yeast PPI simplicial complex.
This observation clearly indicates that increasing the complexity
of the representation of a complex system, i.e., the PPI of yeast,
does not necessarily implies increasing the amount of information
which is extracted from that system. Here we explain these facts based
on the degree of the edges for the sake of simplicity but the explanation
provided is also valid for all the other centralities studied here.
Let us recall that two edges are considered to be adjacent if and
only if they share a node and do not form part of the same triangle.
Thus, the edge degree is given by $k_{p}+k_{q}-(2+2t)$, where $p$
and $q$ are the nodes forming the edge and $t$ is the number of
triangles that the edge is a face of. Notice that in the graph the
edge degree is simply defined as $k_{p}+k_{q}-2$. Now, the important
thing here is that the edge degree in the simplicial complex depends
on the degree of the nodes forming such edge. In Fig. (\ref{plots_Edges})
we show the scatterplots of the edge centrality indices versus their
node analogues. As can be seen in all cases the correlation is positive
and for the cases of the degree and closeness the correlation between
the two centralities is relatively good. We now analyze the causes
of the differences between the node and edge centralities and how
they influence the lack of ability of edge centrality to identify
essential proteins in the yeast PPI.

\begin{figure}
\subfloat[]{\begin{centering}
\includegraphics[width=0.33\textwidth]{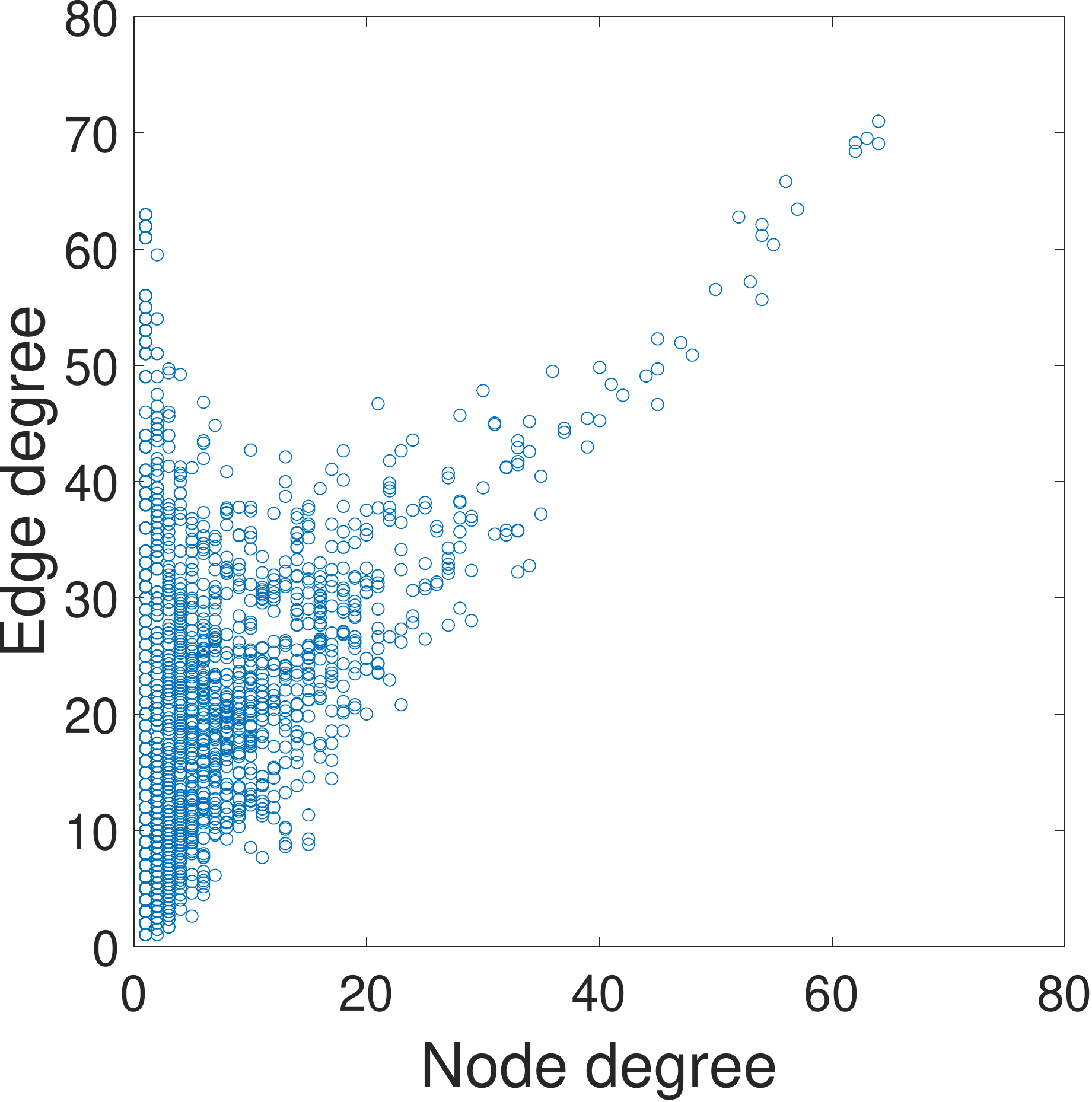}
\par\end{centering}
}\subfloat[]{\begin{centering}
\includegraphics[width=0.33\textwidth]{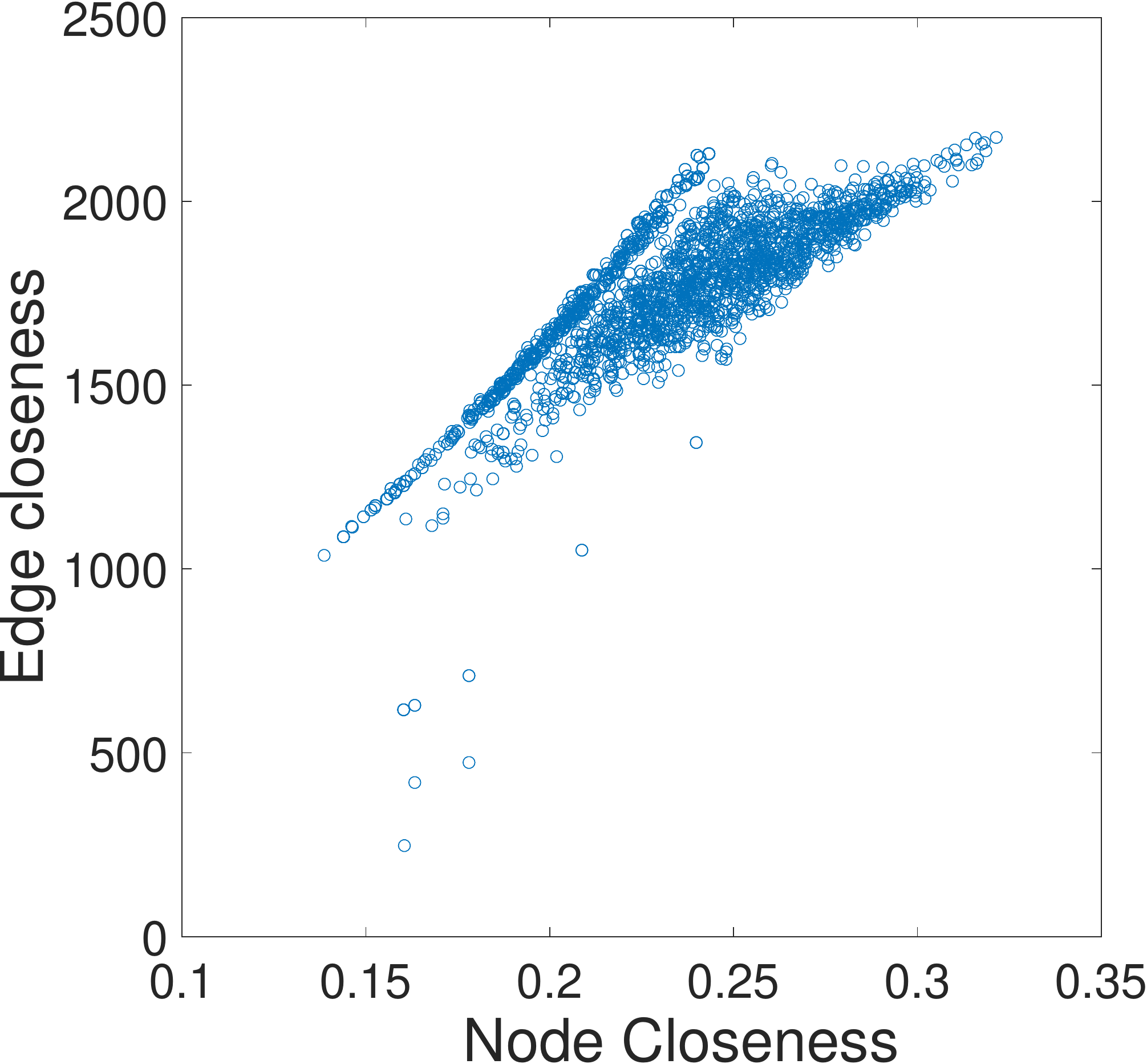}
\par\end{centering}
}\subfloat[]{\begin{centering}
\includegraphics[width=0.31\textwidth]{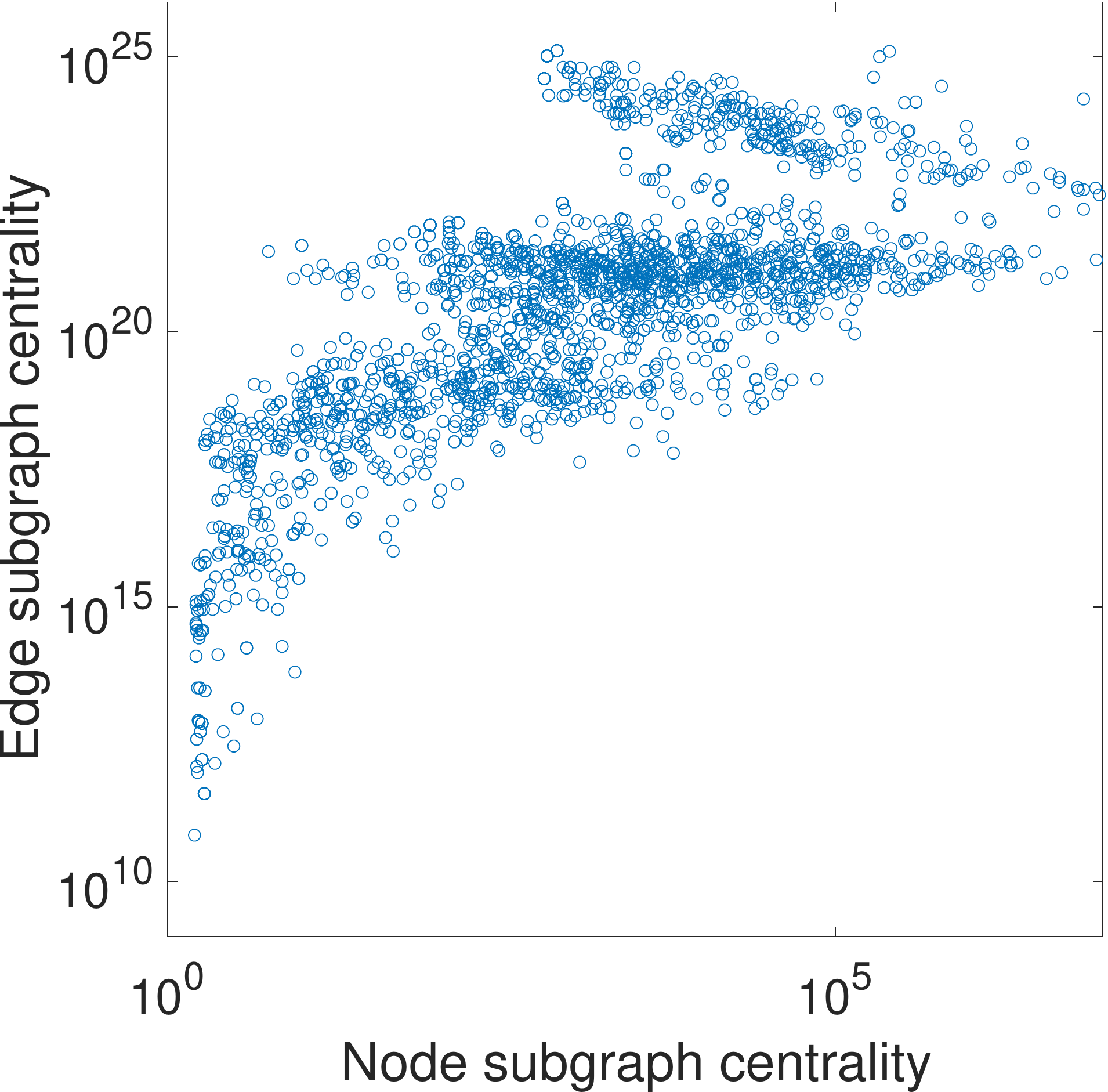}
\par\end{centering}
}

\caption{Scatter plots of the edge centralities, degree (a), closeness (b)
and subgraph centrality (c), versus the node analogues of the same
centralities. Notice that the subgraph centrality plot is in log-log
scale. }

\label{plots_Edges}
\end{figure}

Suppose that the number of triangles that the edge is a face of is
relatively small, such that the degree of the edge is mainly dependent
on the degree of the nodes forming that edge. Then, it is possible
to have two different edges with exactly the same edge degree which
differ significantly in the degree of the nodes forming the edges.
That is, we can have an edge formed by two nodes of mid-degree, e.g.,
MD-MD, and another formed by a high-degree (HD) and a low-degree (LD)
node. It is not difficult to find many of these examples in the yeast
PPI. For instance, the edges YMR125C-YOL139C and YDR386W-YOL139C are
formed by nodes of degrees 39-36 and 31-36, respectively. That is,
these two edges are of the MD-MD type. On the other hand, the edges
YPR110C-YLR086W and YPR110C-YGL016W are formed both by nodes of degrees
64-3, which clearly means that they are of the HD-LD type. It is well-know
that high-degree nodes are more probable to represent essential proteins.
Thus, it is more probable that the edges HD-HD contains an essential
protein than the edges MD-MD. Indeed, neither of the proteins in the
previous example in MD-MD are essential, but the protein YPR110C in
the HD-LD edges is an essential one. Now, the situation is even worse
when the nodes involved in a given edge participate in a large number
of triangles. In this case, twice the number of triangles is subtracted
from the edge degree as we have seen before. Then, if the HD node
involved in an edge is also involved in a large number of triangles,
its edge degree will be relatively small due to the fact that it is
penalized for such participation in triangles. Thus edges which were
HD-HD would be unlikely to have a high edge centrality. The existence
of nodes having low degree but displaying either very low or very
high edge degree is easy to understand. In edges of the HD-LD type,
there is always a node with low degree which displays very large edge
degree due to the influence of the HD node. In those edges where a
LD node is connected to another LD node, both the node and the edge
degree are low. These factors explain very well the failure of all
edge centrality in accounting for the number of essential proteins
in the yeast PPI in a similar way as the node and triangle centralities.

Now we move to the analysis of the triangle centrality indices. The
first thing to be noted is the lack of correlation between the node
and triangle centralities (see Fig. \ref{Node_triangle centralities})
even for the degree centralities. This lack of correlation clearly
indicates that the information contained in one of the indices is
not duplicated by the other. This is confirmed by the fact that the
ranking based on node degree and triangle degree only identify 24
common proteins in the top 100 nodes ranked according to them.

\begin{figure}
\subfloat[]{\begin{centering}
\includegraphics[width=0.33\textwidth]{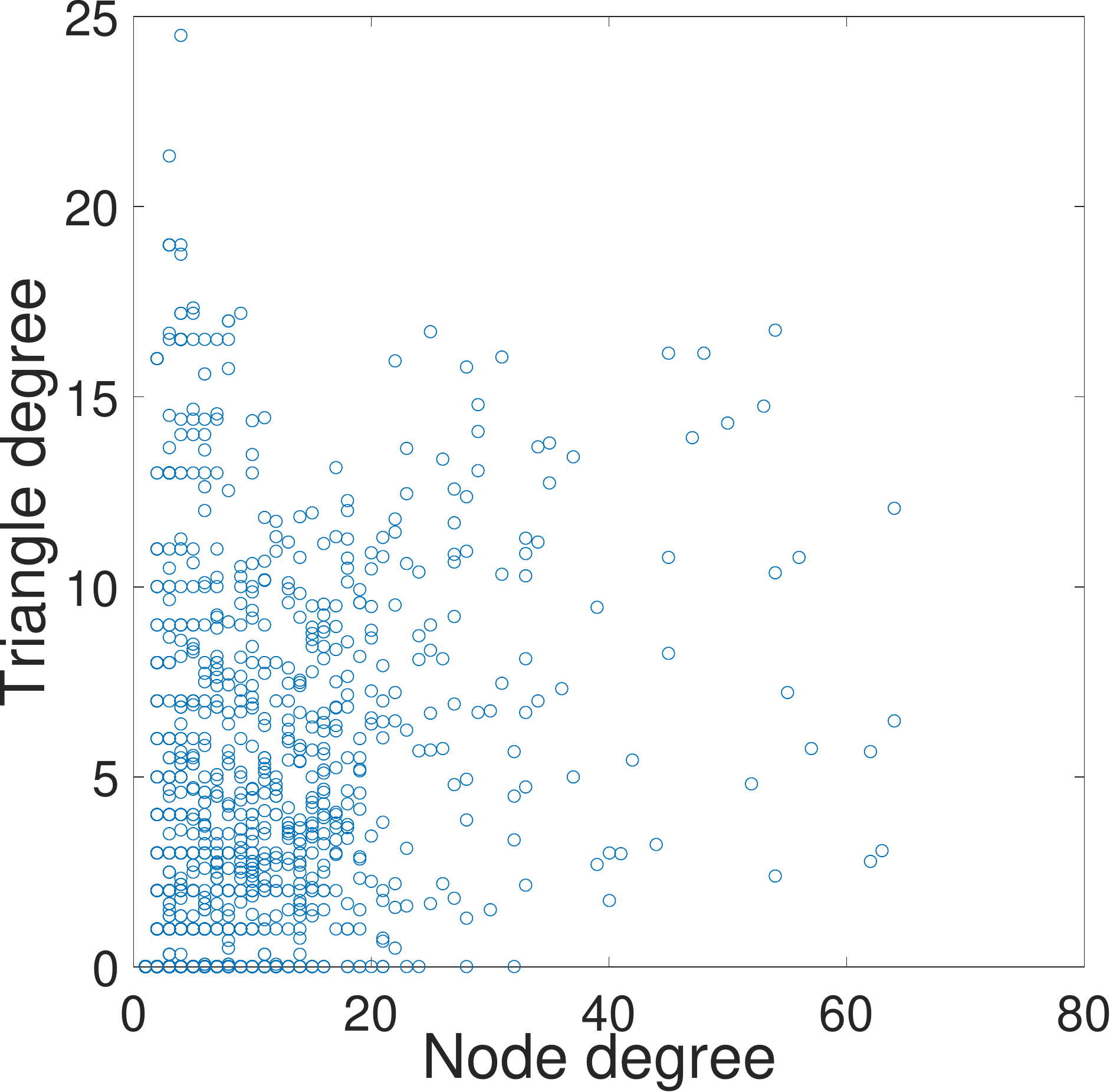}
\par\end{centering}
}\subfloat[]{\begin{centering}
\includegraphics[width=0.33\textwidth]{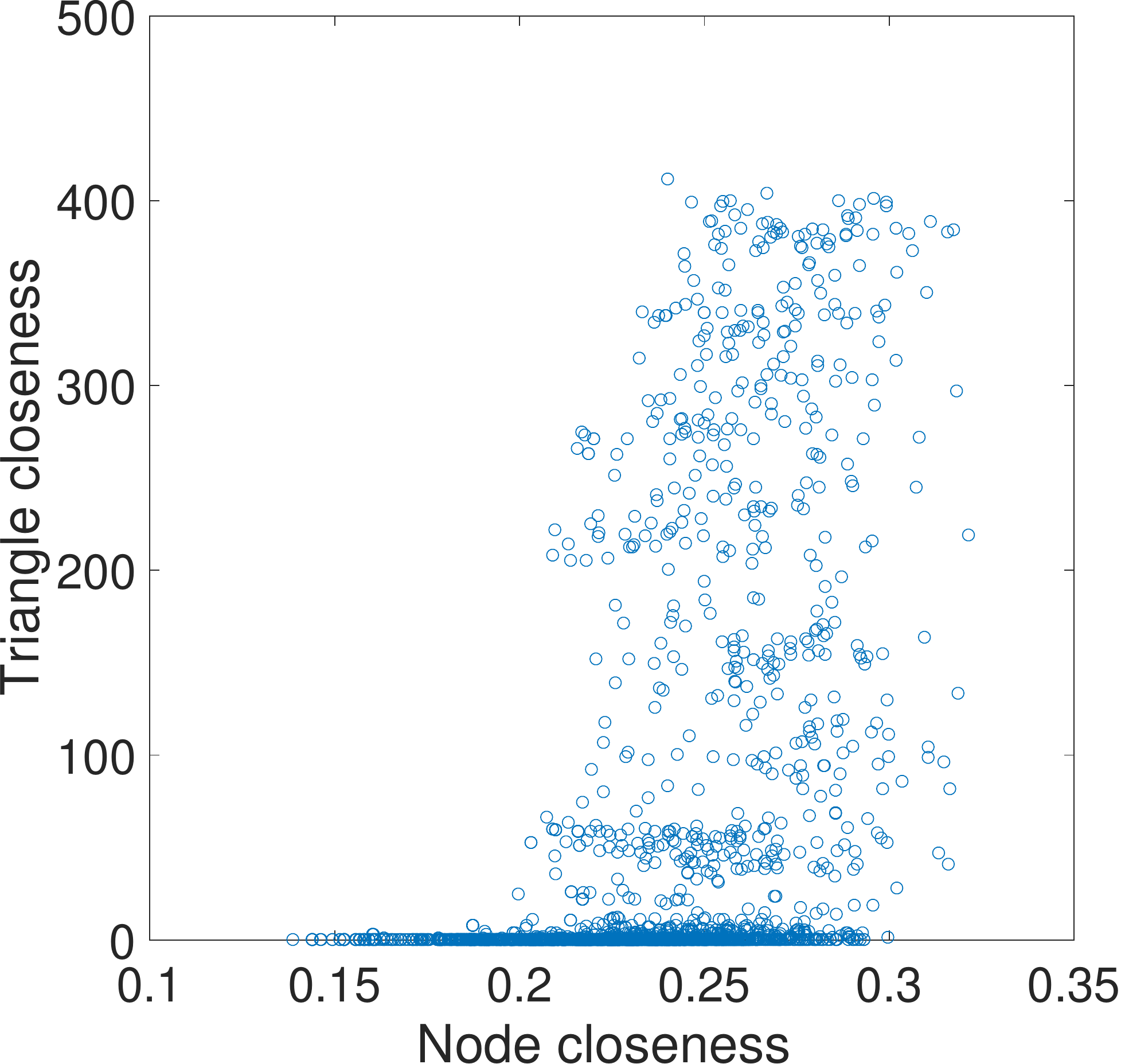}
\par\end{centering}
}\subfloat[]{\begin{centering}
\includegraphics[width=0.33\textwidth]{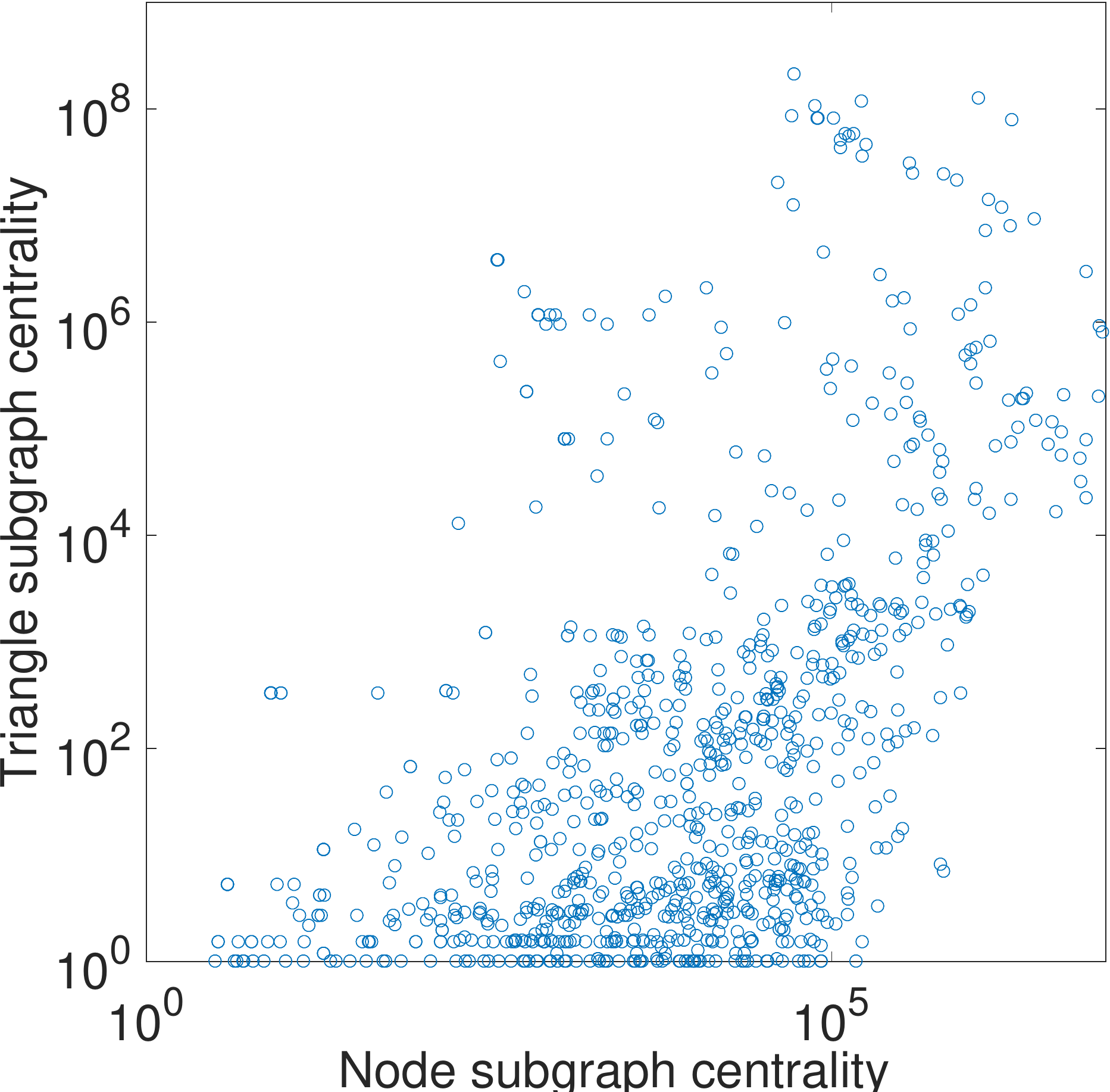}
\par\end{centering}
}

\caption{Scatter plots of the triangle centralities, degree (a), closeness
(b) and subgraph centrality (c), versus the node analogues of the
same centralities. Notice that the subgraph centrality plot is in
log-log scale. }

\label{Node_triangle centralities}
\end{figure}

Because the triangle centralities outperform the node ones in identifying
essential protein, our main goal here is to identify some structural
pattern which contributes to the triangle centralities and do not
contribute to the node ones. In order to perform our analysis we again
consider the degree centralities for the sake of simplicity. We are
only interested in the structural information which is useful for
the identification of essential proteins. The structural pattern that
we identify here consists of a node $A$ which is the vertex of a
relatively small number of triangles, such that its node degree is
small. Suppose for instance that $A$ is connected to the nodes $B$,
$C$, and $D$ forming the triangles $ABC$ and $ACD$. Obviously
the node degree of $A$ is only 3. Now, let us consider that $BC$
is the edge of a large number of triangles, and in a similar way the
edge $CD$. This makes that the triangles $ABC$ and $ACD$ have large
triangle degree and consequently the node $A$ is very central according
to this index. This means that a node is triangle-central not only
if it takes part in a large number of triangles, but also if the edges
of the triangles it forms participate in a large number of triangles.
As a matter of example we provide two complexes displaying exactly
this kind of structural pattern. The first is formed by the protein
YDL148C, which is connected to other proteins, namely YGR090W, YBR247C,
YCL059C and YCR057C. These proteins form 5 triangles in which YDL148C
is a vertex. Then, obviously, the protein YDL148C is not very central
according to this nearest-neighbor structure, i.e., its node degree
is only 4 and it participates in only 5 triangles. However, the edges
of these 5 triangles participate in a total of 88 other triangles.
That is, the edge YGR090W\textendash YBR247C takes place in 11 other
triangles, YGR090W\textendash YCR057C in 22, YBR247C\textendash YCL059C
in 14, YCL059C\textendash YCR057C in 25, and YBR247C\textendash YCR057C
in 16. The protein YDL148C is then very central according to the triangle
centrality. This protein is indeed an essential one. Another example
is provided by the protein YMR112C, which is also essential and is
connected only to YDL005C, YOL135C, YBR253W and YJR068W. It forms
only 4 triangles, but the edges of these triangles form 14, 15, 17
and 19 other triangles, respectively. Thus, the protein YMR112C which
is not central according to node centrality is one of the most central
ones according to the triangle centrality indices.

\begin{figure}
\subfloat[]{\centering{}\includegraphics[width=0.5\textwidth]{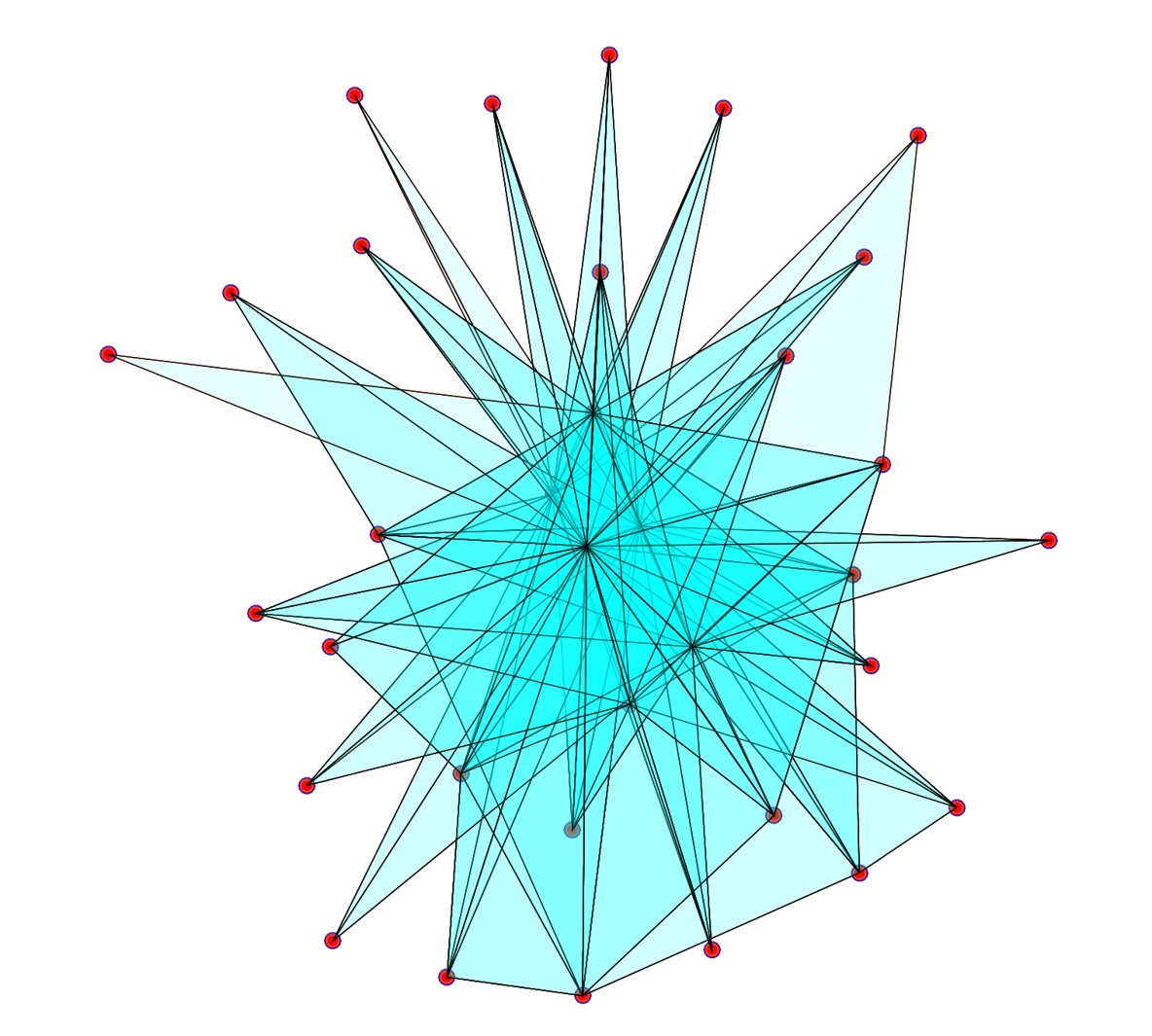}}
\subfloat[]{\centering{}\includegraphics[width=0.5\textwidth]{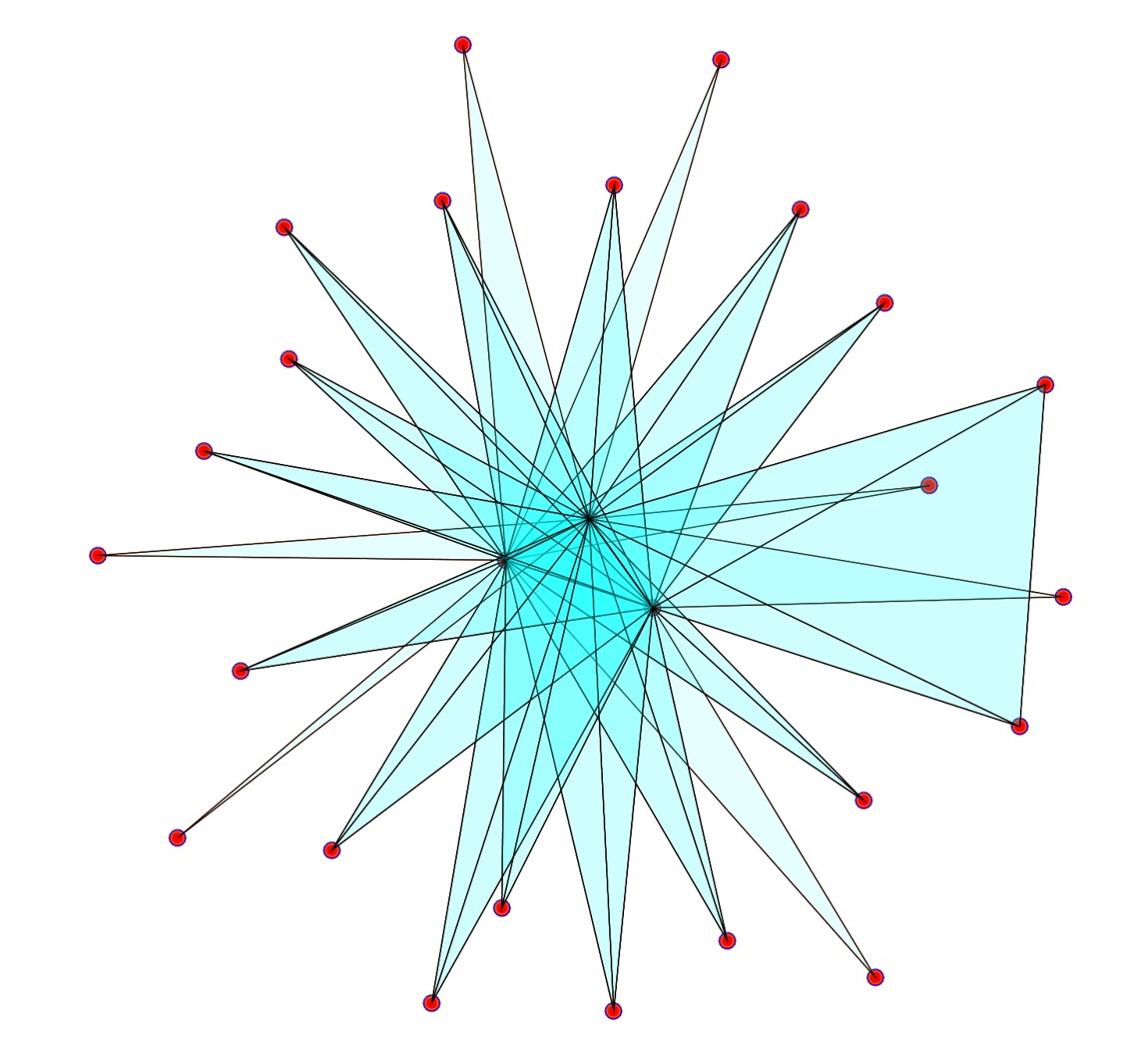}}

\caption{Illustration of the two simplicial complexes in formed by the proteins
YDL148C (a) and YMR112C (b). }
\end{figure}

It should also be noted that there is structural information contained
in the node centralities which is not accounted for by the triangle
ones. As we have seen before there are proteins with high node centrality
and low triangle centrality. However, the number of structural patterns
contributing to this situation is wider and range from the simplest
case when a protein interacts with a large number of other proteins
which do not interact among each other, to the cases in which a central
protein forms a wheel-like structure. In the first case obviously
the protein has a high degree but its triangle degree is zero. In
the second case\textemdash where a central node is connected to every
node of a cycle having $n-1$ nodes, the central node has degree $n-1$
but every triangle has degree only two. In closing, the important
message of this section is the the triangle centrality include some
structural information which is relevant for understanding biological
processes such as the essentiality of proteins in the yeast PPI.

\section{Conclusion}

We have developed here the mathematical framework for the analysis
of centralities in simplicial complexes representing the topological
structure of complex systems. We also provide sufficient examples
on the use of these centrality indices for the analysis of the structure
of simplicial complexes emerging from PPI networks. The main conclusion
of this work is that the understanding of the centrality of simplices
at one level of the simplicial complex, e.g., node centralities, does
not necessarily implies the understanding of the same centrality at
another level, e.g., edges or triangles. This conclusion can have
important consequence for the analysis not only of the structure of
simplicial complexes but also for understanding dynamical processes
taking place on them. We have provided here evidence that has shown
that (i) the ranking of nodes, edges and triangles according to a
given centrality measure can differ significantly for certain simplicial
complexes, (ii) the node, edge, and triangle degrees can display very
significant distributions in some simplicial complexes.We hope that
this work opens these new research avenues for a better understanding
of complex systems.

\section{Acknowledgements}

GR would like to thank the EPSRC for funding his studentship.

\bibliographystyle{apalike}
\bibliography{References1}

\begin{thebibliography}{}

\bibitem[Bavelas, 1950]{bavelas1950communication}
Bavelas, A. (1950).
\newblock Communication patterns in task-oriented groups.
\newblock {\em The Journal of the Acoustical Society of America},
  22(6):725--730.

\bibitem[Bu et~al., 2003]{bu2003topological}
Bu, D., Zhao, Y., Cai, L., Xue, H., Zhu, X., Lu, H., Zhang, J., Sun, S., Ling,
  L., Zhang, N., et~al. (2003).
\newblock Topological structure analysis of the protein--protein interaction
  network in budding yeast.
\newblock {\em Nucleic acids research}, 31(9):2443--2450.

\bibitem[Butland et~al., 2005]{butland2005interaction}
Butland, G., Peregr{\'\i}n-Alvarez, J.~M., Li, J., Yang, W., Yang, X.,
  Canadien, V., Starostine, A., Richards, D., Beattie, B., Krogan, N., et~al.
  (2005).
\newblock Interaction network containing conserved and essential protein
  complexes in escherichia coli.
\newblock {\em Nature}, 433(7025):531--537.

\bibitem[Cang et~al., 2015]{cang2015topological}
Cang, Z., Mu, L., Wu, K., Opron, K., Xia, K., and Wei, G.-W. (2015).
\newblock A topological approach for protein classification.
\newblock {\em Molecular Based Mathematical Biology}, 3(1).

\bibitem[Courtney and Bianconi, 2016]{courtney2016generalized}
Courtney, O.~T. and Bianconi, G. (2016).
\newblock Generalized network structures: The configuration model and the
  canonical ensemble of simplicial complexes.
\newblock {\em Physical Review E}, 93(6):062311.

\bibitem[De~Silva and Ghrist, 2007]{de2007coverage}
De~Silva, V. and Ghrist, R. (2007).
\newblock Coverage in sensor networks via persistent homology.
\newblock {\em Algebraic \& Geometric Topology}, 7(1):339--358.

\bibitem[De~Silva et~al., 2005]{de2005blind}
De~Silva, V., Ghrist, R., and Muhammad, A. (2005).
\newblock Blind swarms for coverage in 2-d.
\newblock In {\em Robotics: Science and Systems}, pages 335--342.

\bibitem[Estrada, 2006]{estrada2006virtual}
Estrada, E. (2006).
\newblock Virtual identification of essential proteins within the protein
  interaction network of yeast.
\newblock {\em Proteomics}, 6(1):35--40.

\bibitem[Estrada, 2010]{estrada2010degreehet}
Estrada, E. (2010).
\newblock Quantifying network heterogeneity.
\newblock {\em Phys. Rev. E}, 82:066102.

\bibitem[Estrada et~al., 2015]{estrada2015first}
Estrada, E., Knight, P.~A., and Knight, P. (2015).
\newblock {\em A first course in network theory}.
\newblock Oxford University Press, USA.

\bibitem[Estrada and Rodriguez-Velazquez, 2005]{estrada2005subgraph}
Estrada, E. and Rodriguez-Velazquez, J.~A. (2005).
\newblock Subgraph centrality in complex networks.
\newblock {\em Physical Review E}, 71(5):056103.

\bibitem[Freeman, 1977]{freeman1977set}
Freeman, L.~C. (1977).
\newblock A set of measures of centrality based on betweenness.
\newblock {\em Sociometry}, pages 35--41.

\bibitem[Ghrist and Muhammad, 2005]{ghrist2005coverage}
Ghrist, R. and Muhammad, A. (2005).
\newblock Coverage and hole-detection in sensor networks via homology.
\newblock In {\em Information Processing in Sensor Networks, 2005. IPSN 2005.
  Fourth International Symposium on}, pages 254--260. IEEE.

\bibitem[Giot et~al., 2003]{giot2003protein}
Giot, L., Bader, J.~S., Brouwer, C., Chaudhuri, A., Kuang, B., Li, Y., Hao, Y.,
  Ooi, C., Godwin, B., Vitols, E., et~al. (2003).
\newblock A protein interaction map of drosophila melanogaster.
\newblock {\em science}, 302(5651):1727--1736.

\bibitem[Giusti et~al., 2016]{giusti2016two}
Giusti, C., Ghrist, R., and Bassett, D.~S. (2016).
\newblock Two's company, three (or more) is a simplex.
\newblock {\em Journal of computational neuroscience}, 41(1):1--14.

\bibitem[Goldberg, 2002]{goldberg2002combinatorial}
Goldberg, T.~E. (2002).
\newblock Combinatorial laplacians of simplicial complexes.
\newblock {\em Senior Thesis, Bard College}.

\bibitem[Gustafson et~al., 2006]{gustafson2006towards}
Gustafson, A.~M., Snitkin, E.~S., Parker, S.~C., DeLisi, C., and Kasif, S.
  (2006).
\newblock Towards the identification of essential genes using targeted genome
  sequencing and comparative analysis.
\newblock {\em Bmc Genomics}, 7(1):265.

\bibitem[Horak et~al., 2009]{horak2009persistent}
Horak, D., Maleti{\'c}, S., and Rajkovi{\'c}, M. (2009).
\newblock Persistent homology of complex networks.
\newblock {\em Journal of Statistical Mechanics: Theory and Experiment},
  2009(03):P03034.

\bibitem[Kass and Raftery, 1995]{kass1995bayes}
Kass, R.~E. and Raftery, A.~E. (1995).
\newblock Bayes factors.
\newblock {\em Journal of the american statistical association},
  90(430):773--795.

\bibitem[Katz, 1953]{katz1953new}
Katz, L. (1953).
\newblock A new status index derived from sociometric analysis.
\newblock {\em Psychometrika}, 18(1):39--43.

\bibitem[Kee et~al., 2016]{kee2016information}
Kee, K.~F., Sparks, L., Struppa, D.~C., Mannucci, M.~A., and Damiano, A.
  (2016).
\newblock Information diffusion, facebook clusters, and the simplicial model of
  social aggregation: a computational simulation of simplicial diffusers for
  community health interventions.
\newblock {\em Health communication}, 31(4):385--399.

\bibitem[Konishi and Kitagawa, 2008]{konishi2008information}
Konishi, S. and Kitagawa, G. (2008).
\newblock {\em Information criteria and statistical modeling}.
\newblock Springer Science \& Business Media.

\bibitem[LaCount et~al., 2005]{lacount2005protein}
LaCount, D.~J., Vignali, M., Chettier, R., Phansalkar, A., et~al. (2005).
\newblock A protein interaction network of the malaria parasite plasmodium
  falciparum.
\newblock {\em Nature}, 438(7064):103.

\bibitem[Lee et~al., 2012]{lee2012persistent}
Lee, H., Kang, H., Chung, M.~K., Kim, B.-N., and Lee, D.~S. (2012).
\newblock Persistent brain network homology from the perspective of dendrogram.
\newblock {\em IEEE transactions on medical imaging}, 31(12):2267--2277.

\bibitem[Lin et~al., 2004]{lin2004hp}
Lin, C.-Y., Chen, C.-L., Cho, C.-S., Wang, L.-M., Chang, C.-M., Chen, P.-Y.,
  Lo, C.-Z., and Hsiung, C.~A. (2004).
\newblock hp-dpi: Helicobacter pylori database of protein interactomes
  embracing-experimental and inferred interactions.
\newblock {\em Bioinformatics}, 21(7):1288--1290.

\bibitem[Maleti{\'c} and Rajkovi{\'c}, 2009]{maletic2009simplicial}
Maleti{\'c}, S. and Rajkovi{\'c}, M. (2009).
\newblock Simplicial complex of opinions on scale-free networks.
\newblock In {\em Complex networks}, pages 127--134. Springer.

\bibitem[Maleti{\'c} and Rajkovi{\'c}, 2012]{maletic2012combinatorial}
Maleti{\'c}, S. and Rajkovi{\'c}, M. (2012).
\newblock Combinatorial laplacian and entropy of simplicial complexes
  associated with complex networks.
\newblock {\em The European Physical Journal Special Topics}, 212(1):77--97.

\bibitem[Maleti{\'c} and Rajkovi{\'c}, 2014]{maletic2014consensus}
Maleti{\'c}, S. and Rajkovi{\'c}, M. (2014).
\newblock Consensus formation on a simplicial complex of opinions.
\newblock {\em Physica A: Statistical Mechanics and its Applications},
  397:111--120.

\bibitem[Motz et~al., 2002]{motz2002elucidation}
Motz, M., Kober, I., Girardot, C., Loeser, E., Bauer, U., Albers, M., Moeckel,
  G., Minch, E., Voss, H., Kilger, C., et~al. (2002).
\newblock Elucidation of an archaeal replication protein network to generate
  enhanced pcr enzymes.
\newblock {\em Journal of Biological Chemistry}, 277(18):16179--16188.

\bibitem[Muhammad and Egerstedt, 2006]{muhammad2006control}
Muhammad, A. and Egerstedt, M. (2006).
\newblock Control using higher order laplacians in network topologies.
\newblock In {\em Proc. of 17th International Symposium on Mathematical Theory
  of Networks and Systems}, pages 1024--1038.

\bibitem[Muhammad and Jadbabaie, 2007]{muhammad2007decentralized}
Muhammad, A. and Jadbabaie, A. (2007).
\newblock Decentralized computation of homology groups in networks by gossip.
\newblock In {\em American Control Conference, 2007. ACC'07}, pages 3438--3443.
  IEEE.

\bibitem[Noirot and Noirot-Gros, 2004]{noirot2004protein}
Noirot, P. and Noirot-Gros, M.-F. (2004).
\newblock Protein interaction networks in bacteria.
\newblock {\em Current opinion in microbiology}, 7(5):505--512.

\bibitem[Petri et~al., 2014]{petri2014homological}
Petri, G., Expert, P., Turkheimer, F., Carhart-Harris, R., Nutt, D., Hellyer,
  P.~J., and Vaccarino, F. (2014).
\newblock Homological scaffolds of brain functional networks.
\newblock {\em Journal of The Royal Society Interface}, 11(101):20140873.

\bibitem[Pirino et~al., 2015]{pirino2015topological}
Pirino, V., Riccomagno, E., Martinoia, S., and Massobrio, P. (2015).
\newblock A topological study of repetitive co-activation networks in in vitro
  cortical assemblies.
\newblock {\em Physical biology}, 12(1):016007.

\bibitem[Rain et~al., 2001]{rain2001protein}
Rain, J.-C., Selig, L., De~Reuse, H., Battaglia, V., et~al. (2001).
\newblock The protein-protein interaction map of helicobacter pylori.
\newblock {\em Nature}, 409(6817):211.

\bibitem[Rochat, 2009]{rochat2009closeness}
Rochat, Y. (2009).
\newblock Closeness centrality extended to unconnected graphs: The harmonic
  centrality index.
\newblock In {\em ASNA}, number EPFL-CONF-200525.

\bibitem[Rual et~al., 2005]{rual2005towards}
Rual, J.-F., Venkatesan, K., Tong, H., Hirozane-Kishikawa, T., et~al. (2005).
\newblock Towards a proteome-scale map of the human protein-protein interaction
  network.
\newblock {\em Nature}, 437(7062):1173.

\bibitem[Seringhaus et~al., 2006]{seringhaus2006predicting}
Seringhaus, M., Paccanaro, A., Borneman, A., Snyder, M., and Gerstein, M.
  (2006).
\newblock Predicting essential genes in fungal genomes.
\newblock {\em Genome research}, 16(9):1126--1135.

\bibitem[Sizemore et~al., 2016]{sizemore2016classification}
Sizemore, A., Giusti, C., and Bassett, D.~S. (2016).
\newblock Classification of weighted networks through mesoscale homological
  features.
\newblock {\em Journal of Complex Networks}, page cnw013.

\bibitem[Stumpf and Ingram, 2005]{stumpf2005probability}
Stumpf, M.~P. and Ingram, P.~J. (2005).
\newblock Probability models for degree distributions of protein interaction
  networks.
\newblock {\em EPL (Europhysics Letters)}, 71(1):152--158.

\bibitem[Symonds and Moussalli, 2011]{symonds2011brief}
Symonds, M.~R. and Moussalli, A. (2011).
\newblock A brief guide to model selection, multimodel inference and model
  averaging in behavioural ecology using akaike's information criterion.
\newblock {\em Behavioral Ecology and Sociobiology}, 65(1):13--21.

\bibitem[Tahbaz-Salehi and Jadbabaie, 2010]{tahbaz2010distributed}
Tahbaz-Salehi, A. and Jadbabaie, A. (2010).
\newblock Distributed coverage verification in sensor networks without location
  information.
\newblock {\em IEEE Transactions on Automatic Control}, 55(8):1837--1849.

\bibitem[Uetz et~al., 2006]{uetz2006herpesviral}
Uetz, P., Dong, Y.-A., Zeretzke, C., Atzler, C., Baiker, A., Berger, B.,
  Rajagopala, S.~V., Roupelieva, M., Rose, D., Fossum, E., et~al. (2006).
\newblock Herpesviral protein networks and their interaction with the human
  proteome.
\newblock {\em Science}, 311(5758):239--242.

\bibitem[Von~Mering et~al., 2002]{von2002comparative}
Von~Mering, C., Krause, R., Snel, B., Cornell, M., et~al. (2002).
\newblock Comparative assessment of large-scale data sets of protein-protein
  interactions.
\newblock {\em Nature}, 417(6887):399.

\bibitem[Wang et~al., 2012]{wang2012identification}
Wang, J., Li, M., Wang, H., and Pan, Y. (2012).
\newblock Identification of essential proteins based on edge clustering
  coefficient.
\newblock {\em IEEE/ACM Transactions on Computational Biology and
  Bioinformatics}, 9(4):1070--1080.

\bibitem[Xia and Wei, 2014]{xia2014persistent}
Xia, K. and Wei, G.-W. (2014).
\newblock Persistent homology analysis of protein structure, flexibility, and
  folding.
\newblock {\em International journal for numerical methods in biomedical
  engineering}, 30(8):814--844.

\bibitem[Xia and Wei, 2015]{xia2015multidimensional}
Xia, K. and Wei, G.-W. (2015).
\newblock Multidimensional persistence in biomolecular data.
\newblock {\em Journal of computational chemistry}, 36(20):1502--1520.

\bibitem[Yu et~al., 2007]{yu2007importance}
Yu, H., Kim, P.~M., Sprecher, E., Trifonov, V., and Gerstein, M. (2007).
\newblock The importance of bottlenecks in protein networks: correlation with
  gene essentiality and expression dynamics.
\newblock {\em PLoS computational biology}, 3(4):e59.

\end{thebibliography}

\pagebreak{}

\section*{Appendix. Degree distributions}
\begin{center}
\begin{longtable}{|>{\centering}m{2cm}|>{\centering}m{9cm}|>{\centering}m{6cm}|}
\hline
Distribution & PDF & Plot\tabularnewline
\hline
\hline
gen-Pareto & $y=\left(\dfrac{1}{\sigma}\right)\left(1+k\dfrac{\left(x-\theta\right)}{\sigma}\right)^{-1-1/k}$

\medskip{}

$k$: tail index (shape) parameter

$\sigma$: scale parameter

$\mu$: threshold (location) parameter & \includegraphics[width=0.33\textwidth]{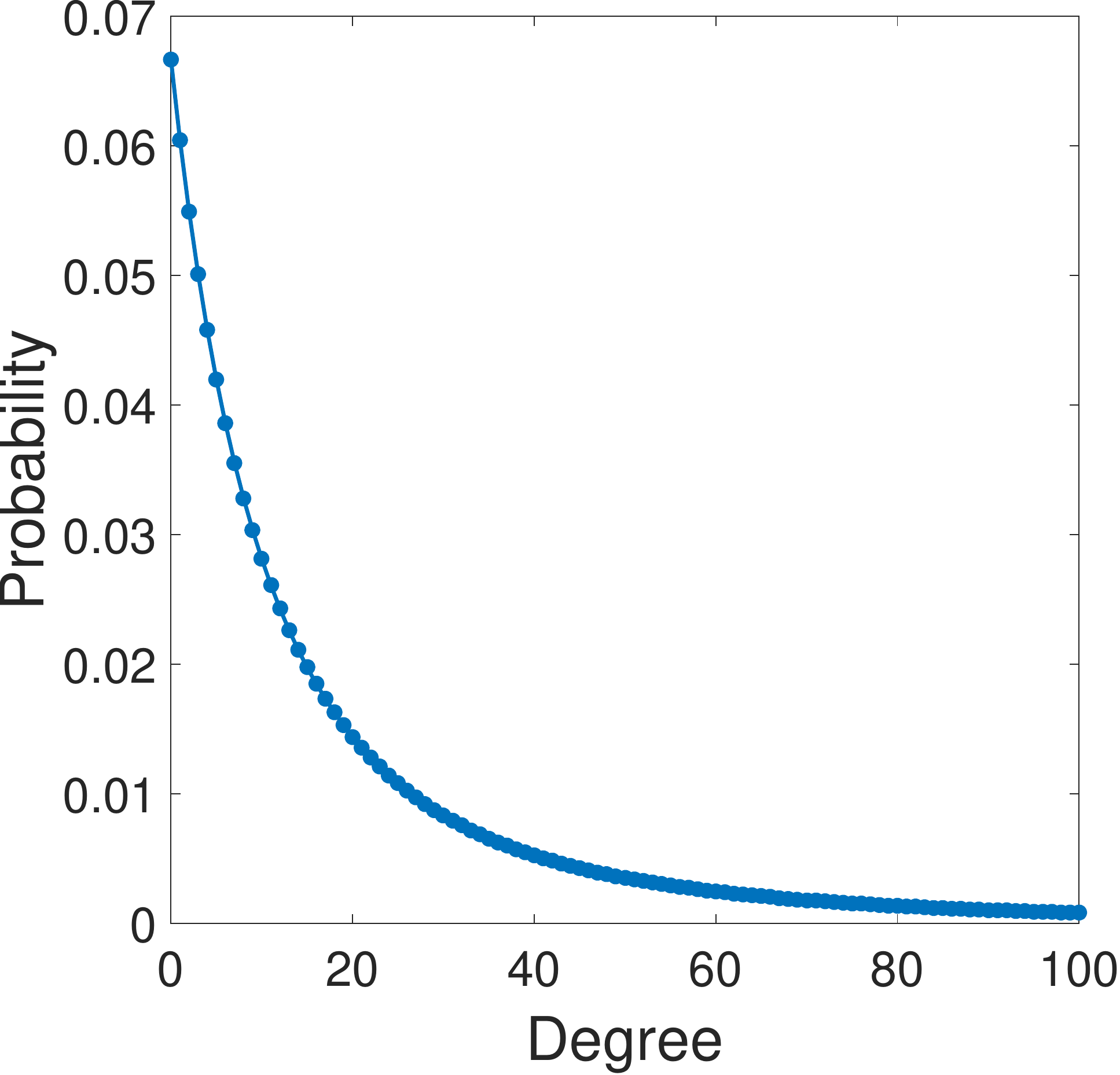}

$k=0.5$, $\sigma=15$, $\mu=0$\tabularnewline
\hline
GVE & $y=\left(\dfrac{1}{\sigma}\right)\exp\left(-\left(1+k\dfrac{\left(x-\mu\right)}{\sigma}\right)^{-1/k}\right)\left(1+k\dfrac{\left(x-\mu\right)}{\sigma}\right)^{-1-1/k}$

\medskip{}

$k$: shape parameter

$\sigma$: scale parameter

$\mu$: location parameter & \includegraphics[width=0.33\textwidth]{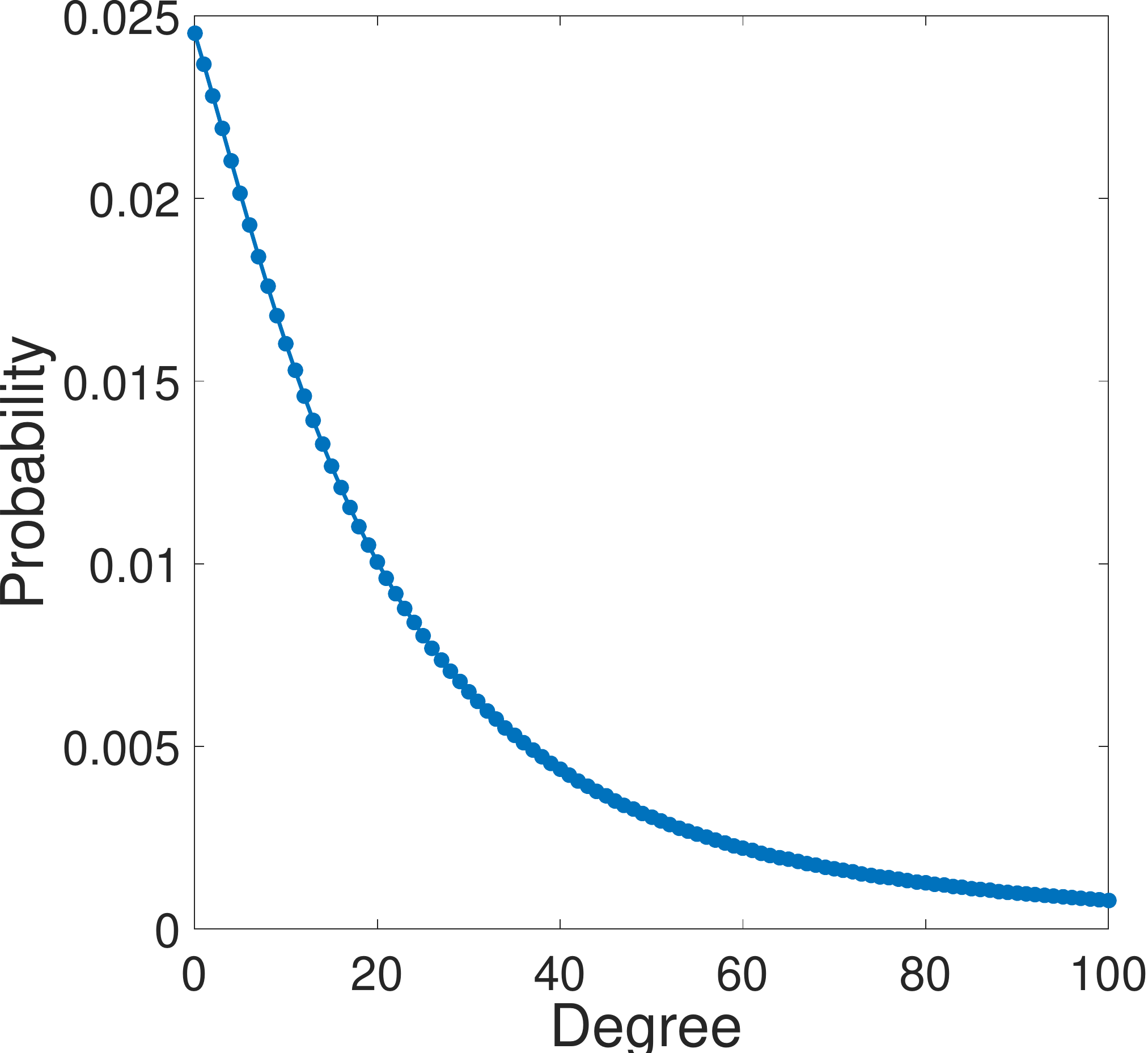}

$k=0.5$, $\sigma=15$, $\mu=0$\tabularnewline
\hline
gamma & $y=\dfrac{x^{a-1}}{b^{a}\Gamma\left(a\right)}\exp\left(-\dfrac{x}{b}\right)$,

\medskip{}

$a$: shape parameter

$b$: scale parameter & \includegraphics[width=0.33\textwidth]{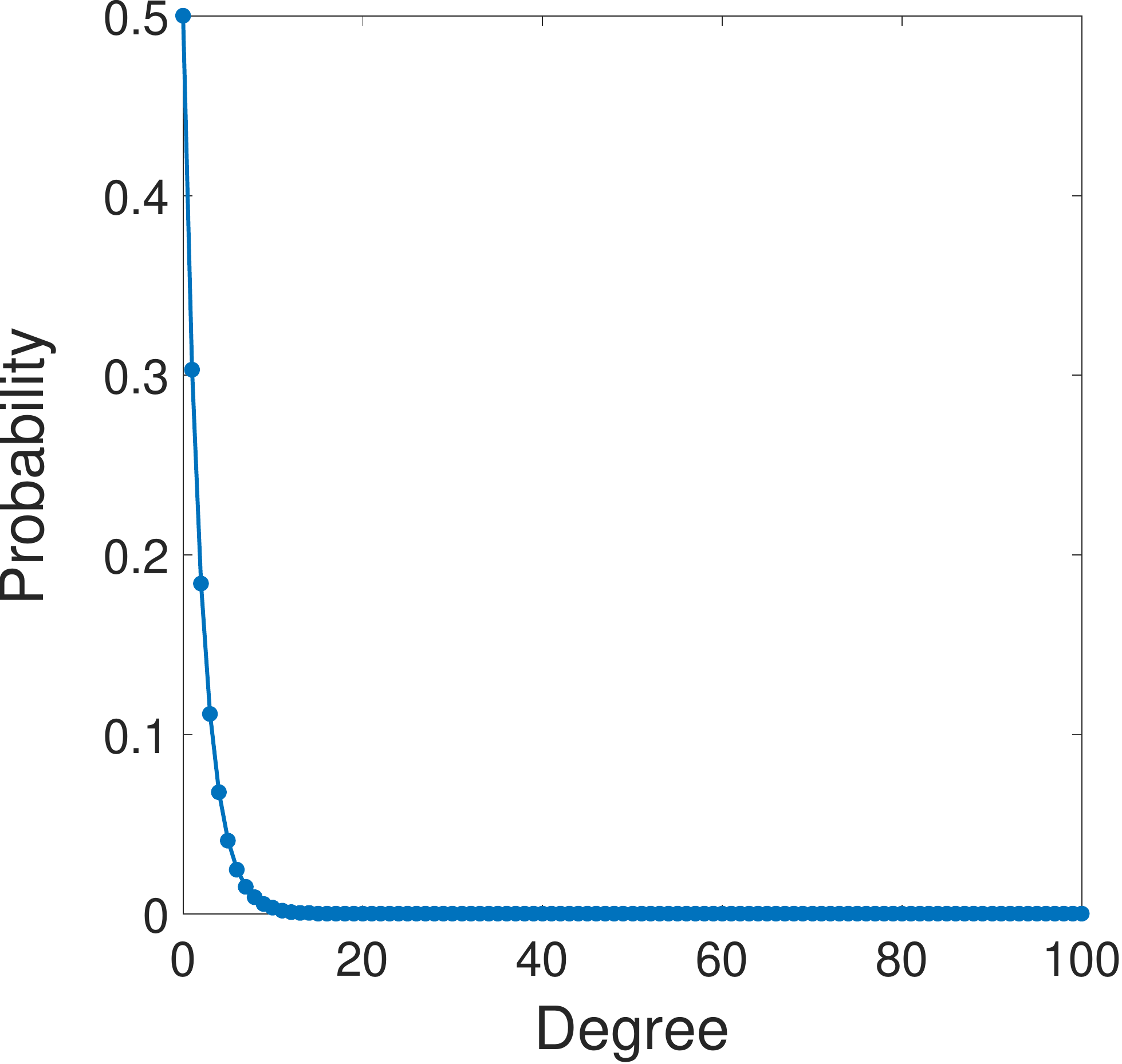}

$a=1$, $b=2$\tabularnewline
\hline
\end{longtable}
\par\end{center}
\end{document}